\documentclass[11pt,a4paper]{article}
\usepackage{amssymb}
\usepackage{amsmath}
\usepackage{amsfonts}
\usepackage{dsfont}
\usepackage{amsthm}
\usepackage{mathrsfs}
\usepackage[hidelinks]{hyperref}
\usepackage{color}
\usepackage[margin=2.41cm]{geometry}
\usepackage[all,cmtip]{xy}
\usepackage[utf8]{inputenc}
\usepackage{graphicx}
\usepackage{varwidth}
\usepackage{comment}
\usepackage{xfrac}

\usepackage{upgreek}
\usepackage{rotating}

\usepackage{tikz}
\usetikzlibrary{shapes.geometric}

\usepackage[shortlabels]{enumitem}

\definecolor{darkred}{rgb}{0.8,0.1,0.1}
\hypersetup{
     colorlinks=false,         
     linkcolor=darkred,
     citecolor=blue,
}

\theoremstyle{plain}
\newtheorem{theo}{Theorem}[section]
\newtheorem{lem}[theo]{Lemma}
\newtheorem{propo}[theo]{Proposition}
\newtheorem{cor}[theo]{Corollary}

\theoremstyle{definition}
\newtheorem{defi}[theo]{Definition}

\newtheorem{exx}[theo]{Example}
\newenvironment{ex}
  {\pushQED{\qed}\exx}
  {\popQED\endexx}

\newtheorem{remm}[theo]{Remark}
\newenvironment{rem}
  {\pushQED{\qed}\remm}
  {\popQED\endremm}
  
\newtheorem{conn}[theo]{Construction}
\newenvironment{con}
  {\pushQED{\qed}\conn}
  {\popQED\endconn}

\numberwithin{equation}{section}

\def\nn{\nonumber}

\def\bbR{\mathbb{R}}
\def\bbC{\mathbb{C}}

\def\bbZ{\mathbb{Z}}
\def\bbT{\mathbb{T}}

\def\Hom{\mathrm{Hom}}

\def\Sym{\mathrm{Sym}}

\def\id{\mathrm{id}}

\def\dd{\mathrm{d}}

\def\cc{\mathrm{c}}
\def\tc{\mathrm{tc}}

\def\1{\mathbf{1}}
\def\oone{\mathds{1}}
\def\op{\mathrm{op}}

\def\pr{\mathrm{pr}}

\def\Loc{\mathbf{Loc}}

\def\Ran{\operatorname{Ran}}

\def\CLoc{\mathbf{CLoc}}
\def\Man{\mathbf{Man}}

\def\AQFT{\mathbf{AQFT}}
\def\2AQFT{\mathbf{2AQFT}}
\def\Fun{\mathbf{Fun}}

\def\Alg{\mathbf{Alg}}

\def\Vec{\mathbf{Vec}}

\def\BB{\mathbf{B}}
\def\CC{\mathbf{C}}
\def\DD{\mathbf{D}}

\def\EE{\mathbf{E}}

\def\TT{\mathbf{T}}
\def\Cat{\mathbf{Cat}}
\def\OCat{\mathbf{OrthCat}}

\def\AAA{\mathfrak{A}}
\def\BBB{\mathfrak{B}}
\def\LLL{\mathfrak{L}}
\def\BBB{\mathfrak{B}}

\def\Sol{\mathfrak{Sol}}
\def\CCR{\mathfrak{CCR}}

\def\O{\mathcal{O}}

\def\As{\mathsf{As}}

\def\skl{\mathrm{skl}}

\def\Mor{\mathrm{Mor}}
\def\Embb{\mathrm{Emb}}
\def\Diff{\mathrm{Diff}}

\def\Mi{\mathbb{M}}
\def\Cy{\sfrac{\mathbb{M}}{\bbZ}}

\newcommand\ovr[1]{\overline{#1}}

\def\sk{\vspace{1mm}}

\makeatletter
\let\@fnsymbol\@alph
\makeatother

\title{%
A skeletal model for $2$d conformal AQFTs
}

\author{%
Marco Benini$^{1,2,a}$, Luca Giorgetti$^{3,b}$\ and\ 
Alexander Schenkel$^{4,c}$\vspace{4mm}\\
{\small ${}^1$ Dipartimento di Matematica, Universit\`a di Genova,}\\
{\small Via Dodecaneso 35, 16146 Genova, Italy.}\vspace{2mm}\\
{\small ${}^2$ INFN, Sezione di Genova,}\\
{\small Via Dodecaneso 33, 16146 Genova, Italy.}\vspace{2mm}\\
{\small ${}^3$ Dipartimento di Matematica, Universit\`a di Roma Tor Vergata,}\\
{\small Via della Ricerca Scientifica 1, 00133 Roma, Italy.}\vspace{2mm}\\
{\small ${}^4$ School of Mathematical Sciences, University of Nottingham,}\\
{\small University Park, Nottingham NG7 2RD, United Kingdom.}\vspace{4mm}\\
{\small \begin{tabular}{ll}
Email: & ${}^a$~\texttt{benini@dima.unige.it}\\
& ${}^b$~\texttt{giorgett@mat.uniroma2.it}\\
& ${}^c$~\texttt{alexander.schenkel@nottingham.ac.uk}\vspace{2mm}
\end{tabular}
}
}

\date{June 2022}

\begin{document}

\maketitle

\vspace{-7mm}

\begin{abstract}
\noindent A simple model for the localization of the category $\mathbf{CLoc}_2$ of oriented and time-oriented globally hyperbolic conformal Lorentzian $2$-manifolds at all Cauchy morphisms is constructed. This provides an equivalent description of $2$-dimensional conformal algebraic quantum field theories (AQFTs) satisfying the time-slice axiom in terms of only two algebras, one for the $2$-dimensional Minkowski spacetime and one for the flat cylinder, together with a suitable action of two copies of the orientation preserving embeddings of oriented $1$-manifolds. The latter result is used to construct adjunctions between the categories of $2$-dimensional and chiral conformal AQFTs whose right adjoints formalize and generalize Rehren's chiral observables.
\end{abstract}

\vspace{-2mm}

\vspace{-2mm}

\vspace{-1mm}

\renewcommand{\baselinestretch}{0.8}\normalsize
\tableofcontents
\renewcommand{\baselinestretch}{1.0}\normalsize

\newpage


\section{\label{sec:intro}Introduction and summary}
This work is a categorical study of $2$-dimensional conformal 
field theories from the perspective of algebraic quantum 
field theory (AQFT). To better illustrate our main results and 
their significance, let us recall the broader context.
$2$-dimensional conformal AQFTs come essentially in two different flavors,
which are often called {\em chiral} and {\em full}. 
Loosely speaking, a chiral theory is sensitive
to only one of the two light-cone coordinates and thus can be
formalized in terms of a net of algebras on a single light ray $\bbR$,
or on its compactification given by the circle $\bbT=\bbR/\bbZ$. We refer the reader
to e.g.\ \cite{KawahigashiLongo,Kawahigashi} for the precise axiomatic
framework for chiral conformal AQFTs and also to \cite{Henriques}
for a natural coordinate-free formulation. In contrast to this,
a full theory is sensitive to both light-cone coordinates
and formalized in terms of a net of algebras on the Minkowski spacetime $\Mi$,
or on its conformal compactification given by the flat cylinder $\Cy$.
See e.g.\ the review article \cite{RehrenReview} for more details.
Such full conformal AQFTs admit an interesting generalization
in the spirit of locally covariant AQFT \cite{BFV,FewsterVerch},
which was studied first by Pinamonti \cite{Pinamonti}. This generalization
treats all conformal spacetimes on the same footing and formalizes 
a theory in terms of a functor $\AAA : \CLoc_2\to \Alg$
from the category $\CLoc_2$ of oriented and time-oriented globally hyperbolic 
conformal Lorentzian $2$-manifolds to a suitable category $\Alg$ of algebras.
(The choice of the category $\Alg$ depends on the context. 
Popular choices are the category of associative and unital 
$\ast$-algebras, the category of $C^\ast$-algebras or the category of von Neumann algebras.
Our main results in this paper hold true for all these cases and many more, see Definition \ref{def:target}.)
Such functor has to satisfy certain axioms, most notably Einstein
causality and the time-slice axiom. Our work is developed in the locally 
covariant framework for $2$-dimensional conformal AQFTs.
\sk

The traditional description of a theory in terms
of a functor $\AAA : \CLoc_2\to \Alg$ satisfying Einstein causality and 
the time-slice axiom is far from being efficient: 
To each of the infinitely many (isomorphism classes of) 
objects $M\in\CLoc_2$ one has to assign an algebra $\AAA(M)\in\Alg$ 
in a functorial way such that Einstein causality and the time-slice axiom hold true.
Note that such assignment is to a large extent redundant, in particular 
due to the time-slice axiom that demands the algebra map
$\AAA(f) : \AAA(M)\to\AAA(M^\prime)$ to be an isomorphism 
for every Cauchy morphism $f:M\to M^\prime$, 
i.e.\ for every (orientation and time-orientation preserving) conformal embedding whose (causally convex) image $f(M)\subseteq M^\prime$ contains
a Cauchy surface of $M^\prime$. One of the main results of this paper
is Theorem \ref{theo:ordinarytominimal} which provides an equivalent
description of $2$-dimensional conformal AQFTs that strips off all the 
redundant data; we call this the {\em skeletal model}.
Our skeletal description of a $2$-dimensional conformal AQFT
is indeed much more efficient than the traditional one as it consists of 
only two algebras, one for the 2-dimensional Minkowski spacetime $\Mi$
and one for the flat cylinder $\Cy$, together with a suitable action 
of two copies of the orientation preserving embeddings of oriented $1$-manifolds.
It is important to emphasize that our result is purely categorical
and hence it does {\em not} rely on specific analytical features 
of the algebras under consideration. In particular, it also holds true
for associative and unital $\ast$-algebras that carry no topology at all.
\sk

The method of proof for Theorem \ref{theo:ordinarytominimal} involves
working in the very broad setting of AQFTs on {\em orthogonal categories} \cite{BSWoperad},
which are general types of spacetime categories that carry information about independent pairs of subsystems,
see Definitions \ref{def:orthcat} and \ref{def:AQFT}. Implementing the time-slice axiom 
can be achieved in this framework by a {\em localization} of orthogonal categories, see Proposition \ref{prop:localization}.
The key result leading to our skeletal model for $2$-dimensional conformal AQFTs
is our explicit and very simple description in Theorem \ref{theo:orthcatrelationship}
of the localization of the category $\CLoc_2$ at all Cauchy morphisms, which crucially 
relies on embedding theorems \cite{Finster,Monclair} from $2$-dimensional conformal
Lorentzian geometry. Hence, our skeletal model seems to be a specific feature
of $2$-dimensional conformal AQFTs that is not likely to generalize to higher dimensions.
\sk

Besides obtaining deep insights into the algebraic structure of $2$-dimensional
conformal AQFTs, our skeletal model is also very useful in applications. As an illustration
of this fact, we study the relationship between $2$-dimensional and chiral conformal AQFTs. 
Our skeletal models both for $2$-dimensional and for chiral conformal AQFTs allow us 
to construct in Theorem \ref{theo:chiralization} two adjunctions
between the category of $2$-dimensional conformal AQFTs
and the category of chiral conformal AQFTs,
whose right adjoints admit an interpretation as `chiralization functors', i.e.\ they extract
the $\pm$-chiral components of a $2$-dimensional conformal AQFT. This construction
is a categorical formalization and also a generalization to the locally covariant setting 
of Rehren's chiral observables \cite{Rehren}.
Furthermore, the following structural result about the relationship
between chiral and full AQFTs is proven: The
category of chiral conformal AQFTs is a {\em full coreflective
subcategory} of the category of $2$-dimensional conformal AQFTs. 
In particular, this entails that the above adjunctions can be used 
to detect which $2$-dimensional conformal AQFTs are chiral. 
\sk

The outline for the remainder of this paper is as follows:
Section \ref{sec:formalism} provides a self-contained review of
the general framework of AQFTs on orthogonal categories \cite{BSWoperad}
in which most of our statements and proofs are written. (Whenever possible, we avoid using operad theory.)
In Section \ref{sec:conformalgeometry} we compute an explicit and very simple description 
of the localization of the category $\CLoc_2$ at all Cauchy morphisms, which culminates in Theorem \ref{theo:orthcatrelationship}.
Section \ref{sec:minmod} provides details for 
how to pass between the (equivalent) ordinary and skeletal descriptions of $2$-dimensional conformal AQFTs.
The passage from ordinary to skeletal is very simple, see Theorem \ref{theo:ordinarytominimal}.
The reconstruction of the ordinary description from the skeletal one always exists but it is computationally
much more involved. Under additional assumptions on the target category $\Alg$, which are satisfied
for $\ast$-algebras but not for $C^\ast$-algebras, 
operadic left Kan extensions provide a concrete model 
for the reconstruction functor, see Theorem \ref{theo:minimaltoordinary}.
In Section \ref{sec:chiralization} we construct chiralization functors that extract
the $\pm$-chiral components of $2$-dimensional conformal AQFTs and investigate their 
categorical properties, see Theorem \ref{theo:chiralization}. Finally, our chiralization functors
are applied to a concrete example in Section \ref{sec:example}, which illustrates how the chiralization of the 
Abelian current is related to the usual chiral currents.
Appendix \ref{app:initial} proves a technical result that is used in Section \ref{sec:chiralization}.


\section{\label{sec:formalism}Orthogonal categories and AQFTs}
We shall briefly recall some relevant definitions and
constructions in algebraic quantum field theory (AQFT).
In order to state and prove the results of this paper,
it will be crucial to work within the general framework of 
AQFTs on orthogonal categories \cite{BSWoperad}. 
Such theories are most elegantly formulated via operads, 
however in order to make our results better accessible to a broader 
audience we shall provide a more elementary and self-contained 
description below.
Our main constructions and results are insensitive 
to specific details of the target category in which the
AQFTs take values. In particular, they hold true 
for $\ast$-algebras, locally convex $\ast$-algebras, 
bornological $\ast$-algebras, convenient $\ast$-algebras,
$C^\ast$-algebras and von Neumann algebras.
In order to avoid excluding cases that might be of interest
to some readers, we work with the following
rather general choice of target category.
\begin{defi}\label{def:target}
We fix once and for all an involutive symmetric monoidal category 
$\TT$ and a (not necessarily full) subcategory $\Alg \subseteq {}^{\ast}\Alg_{\As}(\TT)$
of the category of associative and unital $\ast$-algebras in $\TT$, 
see e.g.\ \cite{Jacobs,BSWinvolutive}.
We assume that the category $\Alg$ is complete, i.e.\ it admits all small limits.
\end{defi}
\begin{ex}
Let us briefly explain how our setting covers the standard choices 
of target categories.
The category ${}^\ast\Alg_{\bbC}=
{}^{\ast}\Alg_{\As}(\Vec_{\bbC})$ of 
associative and unital $\ast$-algebras over $\bbC$  
is obtained by choosing the involutive symmetric 
monoidal category $\TT=\Vec_{\bbC}$ of complex vector spaces.
Choosing instead the involutive symmetric monoidal category 
$\TT=\mathbf{Ban}_{\bbC}$ of Banach spaces, we obtain the complete category 
$C^\ast\Alg_{\bbC}\subseteq {}^{\ast}\Alg_{\As}(\mathbf{Ban}_{\bbC})$
of $C^\ast$-algebras as a full subcategory of the category of Banach $\ast$-algebras.
Furthermore, von Neumann algebras and normal unital $\ast$-homomorphism form a
complete (but not full) subcategory $W^\ast\mathbf{Alg}_\bbC \subseteq C^\ast\Alg_{\bbC}
\subseteq {}^{\ast}\Alg_{\As}(\mathbf{Ban}_{\bbC})$, see e.g.\ \cite{Kornell}.
\end{ex}

The next definition formalizes a general concept of ``spacetime category''
in which certain pairs of morphisms 
$f_1: M_1 \to M^\prime \leftarrow M_2: f_2$ to a common target 
are distinguished. Physically, the distinguished pairs of morphisms should be 
interpreted as ``causally independent subregions'' in $M^\prime$.
\begin{defi}[{\cite[Definition 3.4]{BSWoperad}}]\label{def:orthcat}
An {\em orthogonal category} is a pair $\ovr{\CC} = (\CC,{\perp})$
consisting of a small category $\CC$ and a subset ${\perp} \subseteq 
\Mor\,\CC \,{{}_{\mathsf{t}}{\times}{}_{\mathsf{t}}}\,\Mor\,\CC$
of the set of pairs of morphisms to a common target, such that the 
following conditions hold true:
\begin{itemize}
\item[(i)] Symmetry: If $(f_1,f_2)\in{\perp}$, then $(f_2,f_1)\in{\perp}$.
\item[(ii)] Composition stability: If $(f_1,f_2)\in{\perp}$,
then $(g\,f_1\,h_1 , g\,f_2\,h_2)\in{\perp}$ for all composable
$\CC$-morphisms $g,h_1,h_2$.
\end{itemize}
We often denote elements $(f_1,f_2)\in{\perp}$ by $f_1\perp f_2$.
An {\em orthogonal functor} $F: \ovr{\CC}\to\ovr{\DD}$
is a functor $F : \CC\to\DD$ satisfying $(Ff_1)\perp_{\DD}^{} (Ff_2)$
for all $f_1\perp_\CC^{} f_2$. We denote by $\OCat$ the $2$-category
whose objects are orthogonal categories, morphisms are orthogonal functors
and $2$-morphisms are natural transformations between orthogonal functors.
\end{defi}
\begin{ex}
The prime example $\ovr{\Loc_m} = (\Loc_m,{\perp})$ of an orthogonal category 
is the usual category $\Loc_m$ of oriented and time-oriented
globally hyperbolic Lorentzian manifolds (of a fixed dimension $m\geq 2$) 
with $(f_1:M_1\to M^\prime)\perp(f_2: M_2\to M^\prime)$
if and only if the images $f_1(M_1)$ and $f_2(M_2)$ are causally disjoint 
subsets of $M^\prime$. See e.g.\ \cite{BFV,FewsterVerch} for more details. 
In the present work we will study
various orthogonal categories associated with the 
one of connected oriented and time-oriented globally hyperbolic 
{\em conformal} Lorentzian $2$-manifolds $\ovr{\CLoc_2}$.
The latter was introduced in \cite{Pinamonti} and 
used for conformal field theory studies in \cite{CRV}.
We will give a precise definition of $\ovr{\CLoc_2}$
in Section \ref{sec:conformalgeometry}. 
\end{ex}

The following definition formalizes the concept 
of AQFT on an orthogonal category $\ovr{\CC}$.
\begin{defi}[{\cite[Definition 3.5]{BSWoperad}}]\label{def:AQFT}
Let $\ovr{\CC}$ be an orthogonal category.
An {\em AQFT on $\ovr{\CC}$} is a functor
$\AAA : \CC\to \Alg$ that satisfies the
$\perp$-commutativity property: For all $(f_1:M_1\to M^\prime)\perp(f_2:M_2\to M^\prime)$,
the diagram
\begin{flalign}
\xymatrix@C=5em{
\AAA(M_1)\otimes \AAA(M_2) \ar[r]^-{\AAA(f_1)\otimes \AAA(f_2)}\ar[d]_-{\AAA(f_1)\otimes \AAA(f_2)} & \AAA(M^\prime)\otimes \AAA(M^\prime)\ar[d]^-{\mu_{M^\prime}}\\
\AAA(M^\prime)\otimes \AAA(M^\prime) \ar[r]_-{\mu_{M^\prime}^{\op}}&\AAA(M^\prime)
}
\end{flalign}
in the underlying category $\TT$ commutes, where $\mu_{M^\prime}^{(\op)}$ denotes the (opposite)
multiplication of the algebra $\AAA(M^\prime)$.
We denote by $\AQFT(\ovr{\CC})\subseteq \Fun(\CC,\Alg)$
the full subcategory of $\perp$-commutative functors.
\end{defi}
\begin{rem}
In the case of $\ovr{\Loc_m}$, $\perp$-commutativity is
also called {\em Einstein causality}.
Note that Definition \ref{def:AQFT} does not explicitly 
mention the {\em time-slice axiom} \cite{BFV,FewsterVerch}. We shall show at the end
of this section that the latter
can be implemented through a localization of orthogonal categories.
\end{rem}

The assignment $\ovr{\CC}\mapsto \AQFT(\ovr{\CC})$ 
of the AQFT categories can be promoted to a $2$-functor
\begin{flalign}\label{eqn:AQFT2functor}
\AQFT \,:\,\OCat^{\op} ~\longrightarrow~\Cat\quad,
\end{flalign}
where $\OCat$ is the $2$-category introduced in Definition \ref{def:orthcat} and
$\Cat$ denotes the $2$-category of (not necessarily small) categories.
This $2$-functor assigns to an orthogonal functor $F : \ovr{\CC}\to \ovr{\DD}$ the
pullback functor
\begin{flalign}\label{eqn:pullbackfunctorF}
F^\ast:= (-)F \,:\, \AQFT(\ovr{\DD})~\longrightarrow~ \AQFT(\ovr{\CC})~~,\quad
\AAA~\longmapsto~\AAA\,F
\end{flalign} 
given by pre-composition with $F$. To a natural transformation
$\zeta: F\to G$ between orthogonal functors $F,G:\ovr{\CC}\to\ovr{\DD}$
it assigns the natural transformation
$(-)\zeta : (-)F\to(-)G$ obtained by whiskering.
As every $2$-functor, \eqref{eqn:AQFT2functor} preserves the equivalences
in the respective $2$-categories. The equivalences in $\Cat$ are
the fully faithful and essentially surjective functors
and the equivalences in $\OCat$ can be characterized as follows.
\begin{lem}\label{lem:orth-eq}
An orthogonal functor $F :\ovr{\CC}=(\CC,\perp_{\CC})\to \ovr{\DD}=(\DD,\perp_{\DD})$
is an equivalence in the $2$-category $\OCat$ if and only if the following
two conditions hold true:
\begin{itemize}
\item[1.)] The underlying functor $F :\CC \to\DD$ is fully faithful and essentially surjective.
\item[2.)] The orthogonality relation $\perp_\CC\,=F^\ast(\perp_\DD)$ is the pullback along $F$ 
(see \cite[Lemma 3.19]{BSWoperad})
of the orthogonality relation $\perp_\DD$, i.e.\ $f_1\perp_\CC f_2$ if and only if 
$F(f_1)\perp_\DD F(f_2)$.
\end{itemize}
\end{lem}
\begin{proof}
``$\Leftarrow$'': Since the underlying functor $F: \CC \to \DD$ 
is fully faithful and essentially surjective, it is an equivalence 
in the $2$-category $\Cat$, i.e.\ there exists a functor $G: \DD \to \CC$ 
and natural isomorphisms $\varepsilon: FG \to \id_{\DD}$ and $\nu: \id_{\CC} \to GF$. 
To obtain an equivalence in $\OCat$, we have to prove that
$G : \ovr{\DD}=(\DD,\perp_{\DD})\to \ovr{\CC}= (\CC,\perp_{\CC})$
is an orthogonal functor. Given any orthogonal
pair $(g_1: N_1 \to N^\prime) \perp_{\DD} (g_2: N_2 \to N^\prime)$,
we obtain by using the natural isomorphism $\varepsilon$ that 
$FG(g_i) = \varepsilon_{N^\prime}^{-1}\, g_i\, \varepsilon_{N_i}$, 
for $i=1,2$, hence composition stability of $\perp_{\DD}$ entails that 
the pair $FG(g_1) \perp_{\DD} FG(g_2)$ is orthogonal. 
By definition of $\perp_{\CC}\, = F^\ast(\perp_{\DD})$, it then 
follows that $G(g_1) \perp_{\CC} G(g_2)$ is orthogonal, which implies
that $G$ is an orthogonal functor.
\sk

``$\Rightarrow$'': By hypothesis, there exists an
orthogonal functor $G: \ovr{\DD} \to \ovr{\CC}$ and natural isomorphisms 
$\varepsilon: FG \to \id_{\ovr{\DD}}$ and $\nu: \id_{\ovr{\CC}} \to GF$. This in particular
implies that the underlying functor $F : \CC\to \DD$ 
is fully faithful and essentially surjective.
It remains to check that, 
given $f_i: M_i \to M^\prime$ in $\CC$, for $i=1,2$, 
$F(f_1) \perp_{\DD} F(f_2)$ entails $f_1 \perp_{\CC} f_2$. 
Using the natural isomorphism $\nu$ one finds 
$f_i = \nu_{M^\prime}^{-1}\, GF(f_i)\, \nu_{M_i}$, for $i=1,2$. 
The conclusion follows by recalling that $G$ is an orthogonal functor 
and that $\perp_{\CC}$ is composition stable. 
\end{proof}

In full generality, the pullback functor \eqref{eqn:pullbackfunctorF} does
neither admit a left adjoint nor a right adjoint functor.
The following result provides sufficient conditions for the existence of adjoint functors.
\begin{propo}\label{prop:Kanextension}
Let $F : \ovr{\CC}\to \ovr{\DD}$ be an orthogonal functor.
\begin{itemize}
\item[a)] If the target category $\Alg = {}^\ast\Alg_{\As}(\TT)$ 
is the category of associative and unital $\ast$-algebras in a cocomplete
involutive closed symmetric monoidal category $\TT$, then \eqref{eqn:pullbackfunctorF} 
admits a left adjoint functor
\begin{flalign}
F_!\,:\, \AQFT(\ovr{\CC})~\longrightarrow~\AQFT(\ovr{\DD})\quad.
\end{flalign}
\item[b)] Suppose that the categorical right Kan extension\footnote{Note 
that $\Ran_F$ exists because $\Alg$ was assumed to be complete.}
$\Ran_F : \Fun(\CC,\Alg)\to \Fun(\DD,\Alg)$ preserves $\perp$-commutativity, i.e.\ 
$\Ran_F(\AAA) : \DD\to\Alg$ is $\perp_\DD$-commutative for all
$\perp_\CC$-commutative functors $\AAA : \CC\to\Alg$.
Then the restriction to the AQFT categories
\begin{flalign}\label{eqn:RanF}
F_\ast \,:=\,\Ran_F \,:\, \AQFT(\ovr{\CC})~\longrightarrow~\AQFT(\ovr{\DD})
\end{flalign}
defines a right adjoint functor for \eqref{eqn:pullbackfunctorF}.
\end{itemize}
\end{propo}
\begin{proof}
Item a): Under our hypotheses, the categories $\AQFT(\ovr{\CC})\simeq {}^\ast\Alg_{\O_{\ovr{\CC}}}(\TT)$
and $\AQFT(\ovr{\DD})\simeq {}^\ast\Alg_{\O_{\ovr{\DD}}}(\TT)$ are naturally equivalent
to the categories of $\TT$-valued $\ast$-algebras over the AQFT $\ast$-operads $\O_{\ovr{\CC}}$ and $\O_{\ovr{\DD}}$, see \cite{BSWinvolutive}.
Furthermore, the pullback functor \eqref{eqn:pullbackfunctorF} agrees with
the operadic pullback along $\O_F : \O_{\ovr{\CC}}\to\O_{\ovr{\DD}}$.
Because $\TT$ is by hypothesis cocomplete and closed, 
the $\ast$-operadic left Kan extension \cite[Theorem 6.6]{BSWinvolutive} 
provides the desired left adjoint functor.
\sk

Item b): We have to show that the functor \eqref{eqn:RanF} is right adjoint 
to $F^\ast:  \AQFT(\ovr{\DD})\to \AQFT(\ovr{\CC})$.
For all $\AAA\in \AQFT(\ovr{\CC})$ and $\BBB\in\AQFT(\ovr{\DD})$,
we have a chain of natural isomorphisms
\begin{flalign}
\nn \Hom_{\AQFT(\ovr{\DD})}\big(\BBB, F_\ast(\AAA)\big)\,&=\,
\Hom_{\Fun(\DD,\Alg)}\big(\BBB, \Ran_{F}(\AAA)\big)\\
\nn \,&\cong \,\Hom_{\Fun(\CC,\Alg)}\big(F^\ast(\BBB), \AAA\big)\\
 \,&=\,\Hom_{\AQFT(\ovr{\CC})}\big(F^\ast (\BBB), \AAA\big)\quad.
\end{flalign}
In the first and last step we have used that the AQFT categories
are by Definition \ref{def:AQFT} full subcategories of the functor categories.
The second step follows from the fact that 
$\Ran_F$ is right adjoint to the pullback functor  $F^\ast: 
\Fun(\DD,\Alg)\to \Fun(\CC,\Alg)$.
\end{proof}
\begin{rem}\label{rem:choices}
The hypotheses of item a) are satisfied for the category
${}^\ast\Alg_{\bbC}={}^\ast\Alg_{\As}(\Vec_\bbC)$
of associative and unital $\ast$-algebras over $\bbC$,
but they are {\em not} satisfied for the category $C^\ast\Alg_{\bbC} \subseteq
{}^\ast\Alg(\mathbf{Ban}_{\bbC})$ of $C^\ast$-algebras. 
It is important to emphasize that the left adjoints from item a) 
are {\em not} necessary for proving any of the structural results
of this paper. We will only use them in Section \ref{sec:minmod}
to determine a concrete model for a quasi-inverse of an equivalence
between AQFT categories.
\end{rem}

We conclude this section with a discussion of the {\em time-slice axiom},
which in our general setup from Definition \ref{def:AQFT} takes the following form.
\begin{defi}[{\cite[Definition 3.2]{BSWoperad}}]
Let $\ovr{\CC}$ be an orthogonal category and $W\subseteq \Mor\,\CC$ 
a subset of the set of morphisms. An AQFT $\AAA\in\AQFT(\ovr{\CC})$ is called
{\em $W$-constant} if it assigns to each $(f:M\to M^\prime)\in W$ an isomorphism
$\AAA(f) : \AAA(M)\stackrel{\cong}{\longrightarrow}\AAA(M^\prime)$ in the category $\Alg$.
We denote by $\AQFT(\ovr{\CC})^W\subseteq \AQFT(\ovr{\CC})$ the full subcategory
of $W$-constant AQFTs.
\end{defi}
\begin{ex}
In the case of $\ovr{\Loc_m}$, the subset $W\subseteq \Mor\,\Loc_m$
is chosen to be the set of all \emph{Cauchy morphisms}, i.e.\ $\Loc_m$-morphisms
$f:M\to M^\prime$ such that $f(M)\subseteq M^\prime$ contains a Cauchy surface
of $M^\prime$. Then $W$-constancy reduces to the ordinary time-slice axiom.
\end{ex}

$W$-constant AQFTs on $\ovr{\CC}$ admit an equivalent description in terms of AQFTs
on the {\em localized orthogonal category} $\ovr{\CC[W^{-1}]}$. 
In more detail, let us denote by $\CC[W^{-1}]$ the localization of
the category $\CC$ at the set of morphisms $W$ and by 
$L : \CC\to \CC[W^{-1}]$ the localization functor. (See 
\cite[Section 7.1]{KS} for the relevant definitions.)
We endow $\CC[W^{-1}]$ with the pushforward orthogonality relation
$\perp_{\CC[W^{-1}]}\, := L_\ast(\perp_{\CC})$ (see \cite[Lemma 3.19]{BSWoperad}),
which is the minimal orthogonality relation such that 
$L :  \ovr{\CC}\to \ovr{\CC[W^{-1}]}$ is an orthogonal functor.
\begin{propo}\label{prop:localization}
Let $L : {\ovr{\CC}}\to \ovr{\CC[W^{-1}]} $ be an orthogonal localization functor.
Then the pullback functor
\begin{flalign}\label{eqn:Lpullgeneral}
L^\ast\,:\, \AQFT(\ovr{\CC[W^{-1}]}) ~\longrightarrow~ \AQFT(\ovr{\CC})^W
\end{flalign}
defines an equivalence between the category
of AQFTs on $\ovr{\CC[W^{-1}]}$ and the category of $W$-constant AQFTs on $\ovr{\CC}$.
\end{propo}
\begin{proof}
The pullback functor $L^\ast =(-)L : \AQFT(\ovr{\CC[W^{-1}]}) \to \AQFT(\ovr{\CC})$ 
takes values in the full subcategory $\AQFT(\ovr{\CC})^W\subseteq \AQFT(\ovr{\CC})$
because, by definition, the localization functor $L$ sends morphisms in $W$ to 
isomorphisms. The functor $L^\ast$ is 
fully faithful by definition of localizations (see \cite[Section 7.1]{KS}),
hence it remains to prove that \eqref{eqn:Lpullgeneral} is essentially surjective.
For this let $\BBB\in\AQFT(\ovr{\CC})^W$ be any $W$-constant AQFT on $\ovr{\CC}$.
Again by definition of localizations, there exists a functor $\AAA : \CC[W^{-1}]\to\Alg$
such that $\AAA\,L\cong\BBB$ are naturally isomorphic. 
It is straightforward to check that $\AAA$ is a
$\perp_{\CC[W^{-1}]}$-commutative functor for the pushforward orthogonality relation
because the latter is generated by $L(f_1)\perp_{\CC[W^{-1}]}L(f_2)$ for all
$f_1\perp_{\CC} f_2$. This proves that there exists $\AAA\in \AQFT(\ovr{\CC[W^{-1}]})$
and a natural isomorphism $L^\ast(\AAA)\cong \BBB$, hence \eqref{eqn:Lpullgeneral} 
is essentially surjective.
\end{proof}


\section{\label{sec:conformalgeometry}Localization of $\ovr{\CLoc_2}$ at Cauchy morphisms}
The goal of this section is to develop an explicit
and simple model for the localization
of the orthogonal category $\ovr{\CLoc_2}$
of oriented and time-oriented globally hyperbolic 
{\em conformal} Lorentzian $2$-manifolds at all Cauchy morphisms.
Via Proposition \ref{prop:localization}, this 
will provide us with a very efficient description
of $2$-dimensional conformal AQFTs satisfying the time-slice axiom,
which we will exploit later in the following sections.
\sk

Let us start by recalling the category $\CLoc_2$ and its orthogonal structure, 
see also \cite{Pinamonti,CRV} for earlier appearances of this category.
\begin{defi}
The category $\CLoc_2$ is defined as follows:
Its objects are all oriented and time-oriented Lorentzian $2$-manifolds
$M$ that are globally hyperbolic and connected.\footnote{More pedantically,
one should write $(M,g,\mathfrak{o},\mathfrak{t})$ in order to display the metric $g$, orientation 
$\mathfrak{o}$ and time-orientation $\mathfrak{t}$, but we decided to suppress these data in order
to avoid notational clutter.}
A morphism $f: M\to M^\prime$ is an orientation and time-orientation preserving
embedding with causally convex image $f(M)\subseteq M^\prime$ 
that preserves the conformal structure determined by the metrics, 
i.e.\ $f^\ast g^\prime = \Omega^2 \, g$ for some conformal factor $\Omega^2\in C^\infty(M,\bbR^{>0})$.
The orthogonal category $\ovr{\CLoc_2} := (\CLoc_2,\perp)$ is then defined as follows: 
A pair of morphisms is orthogonal
$(f_1:M_1\to M^\prime)\perp (f_2:M_2\to M^\prime)$ if and only if
the images $f_1(M_1)$ and $f_2(M_2)$ are causally disjoint subsets of $M^\prime$.
Furthermore, a morphism $f: M\to M^\prime$ is called a Cauchy morphism 
if its image $f(M) \subseteq M^\prime$ contains 
a Cauchy surface of $M^\prime$. We denote by 
$W \subseteq \Mor\, \CLoc_{2}$ the set of all Cauchy morphisms. 
\end{defi}
\begin{rem}
We decided to restrict ourselves to connected manifolds in order to simplify
the presentation of this paper. Our results and proofs do generalize 
in a fairly obvious way to non-connected manifolds by treating each connected component separately.
\end{rem}
\begin{ex}
There are two distinguished objects in $\CLoc_2$ that will play a prominent
role in our construction. First, we have the $2$-dimensional 
Minkowski spacetime, which we shall denote by $\Mi\in\CLoc_2$. Explicitly, 
the underlying manifold of $\Mi$ is given by $\bbR^2$, which we describe
by the two light-cone coordinates $x^\pm= t\pm x$.
The metric, orientation and time-orientation then read as 
$g=\tfrac{1}{2}\big(\dd x^+ \otimes \dd x^- + \dd x^-\otimes\dd x^+\big)$,
$\mathfrak{o} = \dd x^-\wedge \dd x^+$ and $\mathfrak{t}=\dd x^{+} + \dd x^{-}$.
Note that every causally convex, connected and open 
subset $U\subseteq \Mi$ defines an object $U\in\CLoc_2$
when endowed with the restricted metric, orientation and time-orientation. Furthermore, the
inclusion map $\iota_U^{\Mi} : U\to \Mi$ is manifestly a $\CLoc_2$-morphism.
\sk

The second distinguished object of interest to us is the $2$-dimensional 
flat cylinder, which we denote by $\Cy\in \CLoc_2$. Explicitly,
the cylinder may be obtained as a quotient of $\Mi$ 
by the $\bbZ$-action $\bbZ\times\Mi \to\Mi\,,~(n,x^\pm) \mapsto x^\pm \pm n$,
together with the induced metric, orientation and time-orientation.
As before, every causally convex, connected and open subset $V\subseteq \Cy$ defines
an object $V\in\CLoc_2$
when endowed with the restricted metric, orientation and time-orientation
and the inclusion map $\iota_V^{\Cy} : V\to \Cy$ is clearly a $\CLoc_2$-morphism.
\end{ex}

Before we can address the problem of localizing $\CLoc_2$ at all Cauchy morphisms,
we have to develop a better understanding of the objects and morphisms in this category.
The first step of our approach consists of using known conformal embedding 
theorems for globally hyperbolic Lorentzian $2$-manifolds in order to obtain 
an equivalent category $\CC_2\simeq \CLoc_2$ that is easier to work with.
For this we recall that, as a consequence of global hyperbolicity and connectedness,
there exist two distinct types of objects $M\in\CLoc_2$: The Cauchy surfaces of $M$ 
are either diffeomorphic to the line $\bbR$ or to the circle $\bbT=\bbR/\bbZ$.
Note that the Minkowski spacetime $\Mi\in\CLoc_2$ is of the first type
and the flat cylinder $\Cy\in\CLoc_2$ is of the second type.
The following embedding results are (essentially, see the proof below) 
proven in \cite[Proposition 4.2]{Finster} and \cite[Theorem 2.2]{Monclair}.
\begin{theo}\label{theo:embedding}
Let $M\in\CLoc_2$ be any object.
\begin{itemize}
\item[a)] If the Cauchy surfaces of $M$ are diffeomorphic to the line $\bbR$, then there exists
a $\CLoc_2$-morphism $M\to \Mi$ into the Minkowski spacetime.

\item[b)] If the Cauchy surfaces of $M$ are diffeomorphic to the circle $\bbT=\bbR/\bbZ$,
then there exists a $\CLoc_2$-morphism $M\to \Cy$ into the flat cylinder. 
This morphism is further a Cauchy morphism, i.e.\ its image contains a
Cauchy surface of $\Cy$.
\end{itemize}
\end{theo}
\begin{proof}
Item a): Ignoring for the moment the orientations and time-orientations, 
it was shown in \cite[Proposition 4.2]{Finster} 
that there exists a conformal embedding $M\to \Mi$ with 
causally convex image. In the case this happens to preserve
the orientations and time-orientations, it defines the desired
$\CLoc_2$-morphism. In the case it does not preserve
the orientation and/or time-orientation, then we post-compose this
embedding with a parity- and/or time-reversal transformation of $\Mi$
to obtain the desired $\CLoc_2$-morphism.
\sk

Item b): It was shown in \cite[Theorem 2.2]{Monclair}
that there exists a conformal embedding $M\to \Cy$. As before,
this embedding can be made orientation and time-orientation preserving
by post-composing (if necessary) with a suitable parity- and/or time-reversal transformation of $\Cy$. 
Picking any space-like Cauchy surface $\Sigma$
of $M$, which by hypothesis is diffeomorphic to the circle 
$\bbT$, its image under the embedding $M\to \Cy$ is a smooth space-like circle in 
$\Cy$, which defines a Cauchy surface $\Sigma^\prime$ of $\Cy$. Hence, the image of $M\to\Cy$
contains a Cauchy surface of $\Cy$. 
\sk

It remains to show that the image
$V\subseteq \Cy$ of the embedding $M\to \Cy$ is causally convex,
which we will prove by contraposition. Suppose that $V\subseteq \Cy$
is not causally convex. Then there exists a causal curve $\gamma : [0,1]\to \Cy$
from $\gamma(0)\in V$ to $\gamma(1)\in V$ that exits and re-enters $V$.
Extending $\gamma$ to an inextensible causal curve $\tilde{\gamma} : \bbR\to \Cy$,
it meets precisely once the space-like Cauchy surface $\Sigma^\prime$ 
in the image $V\subseteq \Cy$ that we have 
constructed above. Co-restricting $\tilde{\gamma}$ to $V\subseteq \Cy$
we obtain an immersion $\tilde{\gamma}\vert : \tilde{\gamma}^{-1}(V)\to V$.
The domain $\tilde{\gamma}^{-1}(V)\subseteq \bbR$ has at least two connected
components, because by hypothesis the causal curve exits and re-enters $V$. By restricting
to the connected components, $\tilde{\gamma}\vert$ defines at least two inextensible causal curves 
in $V$, of which however only one meets the space-like Cauchy surface $\Sigma^\prime$. 
Therefore, inverting the conformal embedding $M \to V\subseteq \Cy$ onto $V$, 
one finds an inextensible causal curve in $M$ that does not meet $\Sigma$. This contradicts 
the hypothesis that $\Sigma$ is a space-like Cauchy surface of $M$. 
\end{proof}

\begin{cor}\label{cor:CLoc2equivalence}
Denote by $\CC_2\subseteq \CLoc_2$ the full subcategory whose
objects are all causally convex, connected and open subsets
$U\subseteq \Mi$ of Minkowski spacetime and all causally convex, connected
and open subsets $V\subseteq \Cy$ of the flat cylinder that contain
a Cauchy surface of $\Cy$. (These subsets are endowed with
the restricted orientation, time-orientation and metric.)
Then the inclusion functor $\CC_2 \to\CLoc_2$ is an equivalence 
of categories that further induces an equivalence of categories 
$\CC_2[W^{-1}] \to\CLoc_2[W^{-1}]$ 
after localization at all Cauchy morphisms.
\end{cor}
\begin{proof}
The inclusion functor is by definition fully faithful.
Essential surjectivity is a direct consequence
of Theorem \ref{theo:embedding} and the fact that
every $\CLoc_2$-morphism $f : M\to M^\prime$ 
can be factorized into a $\CLoc_2$-isomorphism
$M\stackrel{\cong}{\longrightarrow} f(M)$ onto its image,
followed by a subset inclusion $\iota_{f(M)}^{M^\prime} : f(M)\to M^\prime$.
To prove the last part of the statement, let 
$L : \CLoc_{2}\to \CLoc_{2}[W^{-1}]$ be any localization functor.
Using the fact that Cauchy morphisms are closed under composition and 
include all isomorphisms, one easily checks that the composite functor
$\CC_2 \to \CLoc_2 \stackrel{L}{\longrightarrow}\CLoc_2[W^{-1}]$
defines a localization of $\CC_2$ at all Cauchy morphisms.
(This check consists of verifying the three 
conditions in \cite[Definition 7.1.1]{KS} that characterize a localization.)
Uniqueness (up to equivalence of categories) of localizations then implies that
we have an equivalence of categories
$\CC_2[W^{-1}] \stackrel{\sim}{\longrightarrow}\CLoc_2[W^{-1}]$.
\end{proof}

Let us endow $\CC_2 \subseteq \CLoc_2$ with the pullback along the inclusion of the orthogonality
relation $\perp$ on $\CLoc_2$ and denote the resulting
orthogonal category by $\ovr{\CC_2}:= (\CC_2,{\perp})$. 
Combining Lemma \ref{lem:orth-eq} and the equivalence of categories 
from Corollary \ref{cor:CLoc2equivalence}, 
it follows that the inclusion defines an orthogonal 
equivalence $\ovr{\CC_2} \stackrel{\sim}{\longrightarrow} \ovr{\CLoc_2}$. 
Furthermore, a morphism in $\CC_2$ is Cauchy if and only if 
its image in $\CLoc_2$ is such. This entails that 
passing to the orthogonal localizations at all Cauchy morphisms 
both in $\ovr{\CC_2}$ and in $\ovr{\CLoc_2}$ 
defines an equivalence of orthogonal categories 
$\ovr{\CC_2[W^{-1}]}\stackrel{\sim}{\longrightarrow} \ovr{\CLoc_2[W^{-1}]}$. 
Therefore, we can equivalently work with the simpler model 
$\ovr{\CC_2[W^{-1}]}$ instead of $\ovr{\CLoc_2[W^{-1}]}$. 
The next goal is to find an explicit model for the orthogonal
localization $\ovr{\CC_2[W^{-1}]}$ at the subset $W\subseteq\Mor \,\CC_2$ 
of all Cauchy morphisms in $\CC_2$. Our strategy is to construct,
similarly to \cite{BDSboundary}, a reflective localization
by using {\em Cauchy developments} in the ambient spacetimes
$\Mi$ and $\Cy$. Let us recall that, given any object $M\in\CLoc_2$,
the Cauchy development of a subset $S\subseteq M$ 
is the subset $D(S)\subseteq M$ of all points $p\in M$
such that every inextensible causal curve through $p$ meets $S$.
\begin{ex}\label{ex:D}
Let $(V\subseteq\Cy) \in \CC_2$ be a causally convex, connected and open
subset of the flat cylinder that contains a Cauchy surface of $\Cy$.
Then $D(V) = \Cy$ is the full cylinder.
Let now $(U\subseteq \Mi)\in \CC_2$ be a causally convex, connected
and open subset of Minkowski spacetime. Denote
by $\pr_\pm : \Mi=\bbR^2\to \bbR$ the projections onto the light-cone 
coordinates $x^\pm$. Then $D(U) = \pr_+(U)\times \pr_-(U) \subseteq\Mi$
is a double cone in the Minkowski spacetime. The latter is causally convex, connected
and open, hence it defines an object $(D(U)\subseteq \Mi)\in\CC_2$.
\end{ex}

Let us denote by $\CC_2^D \subseteq \CC_2$ the full subcategory
whose objects are stable under Cauchy development. From the example
above (and the standard property $D^2=D$ of Cauchy development), 
we know that there are two kinds of objects in $\CC_2^D$, namely
double cone subsets $U=I_+\times I_-\subseteq \Mi$ of the Minkowski spacetime,
with $I_\pm\subseteq \bbR$ (not necessarily bounded) open intervals, and the full cylinder $\Cy$.
\begin{propo}\label{prop:reflectivelocalization}
The inclusion functor $i : \CC_2^D \to \CC_2$ admits a left adjoint functor
$D : \CC_2\to \CC_2^D$. The latter exhibits $\CC_2^D$
as a reflective localization of $\CC_2$ at all Cauchy morphisms $W\subseteq \Mor\,\CC_2$.
\end{propo}
\begin{proof}
We define the functor $D : \CC_2\to \CC_2^D$
by using Cauchy developments. To objects $(U\subseteq \Mi)\in \CC_2$
and $(V\subseteq \Cy)\in\CC_2$ we assign their Cauchy developments
$D(U) = \pr_+(U)\times \pr_-(U)\subseteq \Mi$ and $D(V) = \Cy$
in the corresponding ambient spacetime. Defining the functor $D$ on morphisms
requires some case distinctions. Consider first the case of a $\CLoc_2$-morphism
$f : (U\subseteq\Mi)\to (U^\prime\subseteq \Mi)$ between two causally convex, connected
and open subsets of Minkowski spacetime. Using light-cone coordinates,
such $f$ takes the form $f(x^+,x^-) = (f^+(x^+,x^-), f^-(x^+,x^-))$
and one easily checks that $f$ is an orientation and time-orientation
preserving conformal embedding if and only if $f(x^+,x^-) = (f^+(x^+),f^-(x^-))$
with $f^\pm: \pr_\pm(U)\to \pr_\pm(U^\prime)$ two orientation preserving embeddings
of intervals. Hence, $f$ canonically extends to a $\CLoc_2$-morphism
$D(f) := f^+\times f^- : 
D(U) = \pr_+(U)\times \pr_-(U) \to D(U^\prime) = \pr_+(U^\prime)\times \pr_-(U^\prime)$
between the Cauchy developments.
\sk

Next, let us consider the case of a $\CLoc_2$-morphism of the form $f : (U\subseteq\Mi)\to (V\subseteq \Cy)$.
Causal convexity implies that the image $f(U)\subseteq V$ is contained in a double cone
subset of the cylinder, hence the same argument as above applies
and we obtain a canonical extension of $f$ to a $\CLoc_2$-morphism
$D(f) := f^+\times f^- : D(U)=\pr_+(U)\times\pr_-(U)\to D(V)= \Cy$.
\sk

As there are no $\CLoc_2$-morphisms of the form $f : (V\subseteq \Cy)\to (U\subseteq\Mi)$
(recall that $V$ contains a Cauchy surface of $\Cy$),
we are left with the remaining case of a $\CLoc_2$-morphism 
$f : (V\subseteq\Cy)\to (V^\prime\subseteq\Cy)$.
Since the quotient map $q : \Mi\to\Cy$ is a universal cover
of $\Cy$ and both $V,V^\prime\subseteq\Cy$ contain a Cauchy surface
of $\Cy$, the preimages $\tilde{V}:= q^{-1}(V)\subseteq \Mi$ and $\tilde{V}^\prime:=
q^{-1}(V^\prime)\subseteq \Mi$ define universal covers of, respectively, $V$ and $V^\prime$.
The conformal embedding $f : V\to V^\prime$ lifts to a conformal immersion
$\tilde{f} : \tilde{V}\to \tilde{V}^\prime$ between the universal covers.
One easily checks that $\tilde{f}$ is of the form $\tilde{f}(x^+,x^-) = (\tilde{f}^+(x^+),
\tilde{f}^-(x^-))$ with $\tilde{f}^\pm: \pr_\pm(\tilde{V})=\bbR \to \pr_\pm(\tilde{V}^\prime)=\bbR$ 
two orientation preserving embeddings satisfying the $\bbZ$-equivariance
condition $\tilde{f}^\pm(y+1) =\tilde{f}^\pm(y) +1$, for all $y\in \bbR$.
Passing to the quotients then defines the desired $\CLoc_2$-morphism
$D(f) := \tilde{f}^+\times\tilde{f}^- : \Cy\to \Cy$.
\sk

The resulting functor $D : \CC_2 \to \CC_2^D$ is left adjoint to the inclusion
functor $i : \CC_2^D\to \CC_2$. The adjunction unit $\eta : \id_{\CC_2} \to i\,D$
is given by the components $\eta_{U}:=\iota_{U}^{D(U)} :U\to D(U)$, for all $(U\subseteq\Mi)\in\CC_2$,
and $\eta_V := \iota_V^{D(V)} : V\to D(V)$, for all $(V\subseteq\Cy)\in\CC_2$.
The adjunction counit $\epsilon : D\,i\to \id_{\CC_2^D}$ is given by the components
$\epsilon_{U} := \id_U :  D(U)=U\to U$, for all $(U\subseteq\Mi)\in\CC^D_2$, 
and $\epsilon_{\Cy}:=\id_{\Cy} : D(\Cy)=\Cy\to \Cy$.
The proof that $D : \CC_2 \to \CC_2^D$ is a localization functor for the localization of $\CC_2$ 
at all Cauchy morphisms is completely analogous to the proof in \cite[Proposition 3.3]{BDSboundary}.
Alternatively, one can also see this more abstractly: Observing that 
$D(f)$ is an isomorphism if and only if $f$ is a Cauchy morphism,
the claim follows from the general theory of reflective localizations.
\end{proof}

We endow the localization $\CC_2^D$ with the pushforward 
orthogonality relation $D_\ast(\perp)$ (see \cite[Lemma 3.19]{BSWoperad}),
which in the present case coincides with the pullback of
$\perp$ along the inclusion functor $i : \CC_2^D\to \CC_2$. 
(This follows from the fact that two subsets are causally disjoint if and only if
their Cauchy developments are.) In simpler words,
two morphisms to a common target are orthogonal in $\ovr{\CC_2^D}$ if and only if their 
images are causally disjoint subsets in the target. 
We denote by $\CC_2^{D,\skl}\subseteq \CC_2^D$ the full subcategory
with only two objects, the Minkowski spacetime $\Mi$ and the flat cylinder $\Cy$,
and endow it with the pullback along the inclusion of the orthogonality relation of $\ovr{\CC_2^D}$.
\begin{propo}
The full subcategory inclusion $\CC_2^{D,\skl}\to \CC_2^D$ defines
an equivalence $\ovr{\CC_2^{D,\skl}}\to \ovr{\CC_2^D}$ of orthogonal categories.
\end{propo}
\begin{proof}
The inclusion functor is by definition fully faithful.
To prove essential surjectivity, recall from Example \ref{ex:D}
that the objects in $\CC_2^D$ are either the flat cylinder $\Cy$
or double cone subsets $U = I_+ \times I_-\subseteq\Mi$ in the Minkowski spacetime.
Hence, essential surjectivity would follow if we could
prove that for each $U = I_+ \times I_-\subseteq\Mi$  there exists
a $\CLoc_2$-isomorphism $f : U\stackrel{\cong}{\longrightarrow} \Mi$ to the full Minkowski spacetime $\Mi$.
Using the characterization of these morphisms from the proof of
Proposition \ref{prop:reflectivelocalization}, we see that this is indeed the case:
Simply take any two orientation preserving diffeomorphisms $f^\pm : I_\pm 
\stackrel{\cong}{\longrightarrow} \bbR$ onto the real line.  
Recalling also that by definition the orthogonality relation of 
$\ovr{\CC_2^{D,\skl}}$ is the pullback along the inclusion 
of the one of $\ovr{\CC_2^D}$, the claim follows by Lemma \ref{lem:orth-eq}. 
\end{proof}

Let us summarize the main result of this section in a useful
diagrammatic form.
\begin{theo}\label{theo:orthcatrelationship}
The various orthogonal categories constructed in this section
are related by the following diagram of orthogonal functors
\begin{flalign}\label{eqn:OCatzigzag}
\xymatrix@C=1.2em@R=1.2em{
\ovr{\CC_2^{D,\skl}} \ar[rr]^-{\sim}~&&~ \ar[ddrr]_-{\id_{\ovr{\CC_2^D}}}\ovr{\CC_2^D} \ar[rr]^-{i} ~&&~\ar@{=>}[ld]_-{\epsilon}^{\cong} \ar@{=>}[rd]^-{\cong} \ovr{\CC_2}\ar[dd]|-{D^{}} \ar[rr]^-{\sim}~&&~\ovr{\CLoc_2}\ar[ddll]^-{L}\\
~&&~ ~&&~ ~&&~\\
~&&~ ~&&~\ovr{\CC_2^D} ~&&~
}
\end{flalign}
that commutes up to the displayed natural isomorphisms.
In this diagram equivalences of orthogonal categories are labeled by $\sim$.
The symbol $\epsilon$ denotes the counit of the reflective localization 
$D\dashv i$ from Proposition \ref{prop:reflectivelocalization} and
$L : \ovr{\CLoc_2} \stackrel{\sim}{\longrightarrow} \ovr{\CC_2}\stackrel{D}{\longrightarrow}
\ovr{\CC_2^D}$ is the orthogonal localization functor that is obtained
by choosing any quasi-inverse of the equivalence $\ovr{\CC_2} \stackrel{\sim}{\longrightarrow} \ovr{\CLoc_2}$.
\end{theo}

We conclude this section by describing the orthogonal category
$\ovr{\CC_2^{D,\skl}}$ more explicitly. By definition,
this category has only two objects, the Minkowski spacetime $\Mi$
and the flat cylinder $\Cy$.
Using similar arguments as in the proof of Proposition \ref{prop:reflectivelocalization},
we can also describe the corresponding Hom-sets. 
Denote by $\Embb^+(\bbR)$ the set of orientation 
preserving embeddings of $\bbR$ into itself. 
For the endomorphisms of Minkowski spacetime, we find a bijection
\begin{subequations}\label{eqn:explicitmorphisms}
\begin{flalign}
\Hom_{\CLoc_2}^{}\big(\Mi,\Mi\big)\,\cong\, \Embb^+(\bbR)^2\quad.
\end{flalign}
Explicitly, the $\CLoc_2$-morphism associated with 
a pair $(f^+,f^-)\in\Embb^+(\bbR)^2$ of orientation preserving embeddings
of $\bbR$ into itself reads as $f: \Mi\to\Mi\,,~(x^+,x^-)\mapsto (f^+(x^+),f^-(x^-))$.
Furthermore, denote by $\Diff^+(\bbT)$ the set of 
orientation preserving diffeomorphisms of $\bbT=\bbR/\bbZ$. 
For the endomorphisms of the flat cylinder, we find the bijection
\begin{flalign}
\Hom_{\CLoc_2}^{}\big(\Cy,\Cy\big)\,\cong\, \Diff^+(\bbT)^2
\end{flalign}
given by associating to a pair $(g^+,g^-)\in \Diff^+(\bbT)^2$
of orientation preserving diffeomorphisms of $\bbT=\bbR/\bbZ$
the $\CLoc_2$-automorphism $g : \Cy\to\Cy\,,~[(x^+,x^-)]\mapsto [(g^+(x^+),g^-(x^-))]$.
Lastly, denote by $\Embb^{+,\leq 1}(\bbR)$ the set of 
orientation preserving embeddings of $\bbR$ into itself 
whose image is an open interval of length $\leq 1$. 
For the mixed Hom-sets, we find
\begin{flalign}
\Hom_{\CLoc_2}^{}\big(\Cy,\Mi\big)\,=\,\emptyset
\end{flalign}
and a bijection
\begin{flalign}
\Hom_{\CLoc_2}^{}\big(\Mi,\Cy\big)\,\cong\, \Embb^{+,\leq 1}(\bbR)^2\big/\bbZ\quad.
\end{flalign}
\end{subequations}
Here the $\bbZ$-action 
on $\Embb^{+,\leq 1}(\bbR)^2$ is given by translation $(f^+,f^-)\mapsto (f^+ +n , f^- -n)$,
for all $n\in \bbZ$. The $\CLoc_{2}$-morphism associated with
$[f^+,f^-]\in\Embb^{+,\leq 1}(\bbR)^2\big/\bbZ $ reads
as $f: \Mi\to\Cy\,,~(x^+,x^-)\mapsto [(f^+(x^+),f^-(x^-))]$.
\sk

To characterize the orthogonality relation on $\ovr{\CC_2^{D,\skl}}$,
let us first note that the causal future/past of a (possibly unbounded) double cone subset
$I_+\times I_- := (a^+,b^+)\times (a^-,b^-)\subseteq \Mi$ is given by
\begin{flalign}
\nn J^{+}_{\Mi}(I_+\times I_-) \,&=\, \big\{ (x^+,x^-)\in\Mi \,:\, 
x^+ > a^+ \text{ and }x^- > a^- \big\}\,\subseteq\,\Mi\quad,\\
J^{-}_{\Mi}(I_+\times I_-) \,&=\, \big\{ (x^+,x^-)\in\Mi \,:\, 
x^+ < b^+ \text{ and }x^- < b^- \big\}\,\subseteq\,\Mi\quad.
\end{flalign}
As usual, we denote their union by
$J_{\Mi}(I_+\times I_-):=J^+_{\Mi}(I_+\times I_-)\cup J^-_{\Mi}(I_+\times I_-)\subseteq \Mi$.
Then the orthogonality relation on $\ovr{\CC_2^{D,\skl}}$ is given explicitly
as follows:
\begin{itemize}
\item[(i)] $((f_1^+,f_1^-):\Mi\to\Mi)\perp((f_2^+,f_2^-):\Mi\to\Mi)$ are orthogonal if and only if
$J_{\Mi}\big(f_1^+(\bbR)\times f_1^-(\bbR)\big)\cap \big(f_2^+(\bbR)\times f_2^-(\bbR)\big)=\emptyset$.

\item[(ii)] $([f_1^+,f_1^-]:\Mi\to\Cy)\perp([f_2^+,f_2^-]:\Mi\to\Cy)$ are orthogonal
if and only if, for all $n \in \bbZ$, 
$J_{\Mi}\big(f_1^+(\bbR)\times f_1^-(\bbR)\big)\cap 
\big((f_2^+(\bbR)+n)\times (f_2^-(\bbR)-n)\big)=\emptyset$.

\item[(iii)] $(g^+,g^-) : \Cy\to\Cy$ is not orthogonal to any morphism.
\end{itemize}


\section{\label{sec:minmod}Skeletal model and reconstruction} 
As a consequence of Theorem \ref{theo:orthcatrelationship} and Proposition \ref{prop:localization}, 
the two-object orthogonal category $\ovr{\CC_2^{D,\skl}}$ captures
the theory of $2$-dimensional conformal AQFTs satisfying the time-slice axiom,
in the sense that we have an equivalence of categories $\AQFT(\ovr{\CLoc_2})^W \simeq
\AQFT(\ovr{\CC_2^{D,\skl}})$. (See Theorem \ref{theo:ordinarytominimal} below for
the precise statement.) 

The latter perspective is very efficient:
By Definition \ref{def:AQFT}, a theory $\AAA\in \AQFT(\ovr{\CC_2^{D,\skl}})$ simply consists
of two algebras, one for the Minkowski spacetime $\AAA(\Mi)\in\Alg$
and one for the flat cylinder $\AAA(\Cy)\in\Alg$,
together with a $\perp$-commutative action of the morphisms in $\ovr{\CC_2^{D,\skl}}$.
We refer to this description of $2$-dimensional conformal 
AQFTs satisfying time-slice as the {\em skeletal model}.
\sk

The aim of this section is to spell out in more detail how to pass between
the ordinary description and the skeletal one. Composing the horizontal 
orthogonal functors in the diagram \eqref{eqn:OCatzigzag} of 
Theorem \ref{theo:orthcatrelationship}, we obtain the orthogonal
full subcategory inclusion
\begin{flalign}\label{eqn:jinclusion}
j\,:\,\ovr{\CC_2^{D,\skl}} ~\longrightarrow~\ovr{\CLoc_2}\quad.
\end{flalign}
\begin{theo}\label{theo:ordinarytominimal}
The restriction of the pullback functor
\begin{flalign}\label{eqn:jpull}
j^\ast\,:\, \AQFT(\ovr{\CLoc_2})^W~\longrightarrow~\AQFT(\ovr{\CC_2^{D,\skl}})
\end{flalign}
to the full subcategory of $2$-dimensional conformal AQFTs satisfying time-slice is an equivalence of categories.
\end{theo}
\begin{proof}
Observe that the orthogonal functor $j$ defined in \eqref{eqn:jinclusion}
is the composition of the horizontal orthogonal functors in
\eqref{eqn:OCatzigzag}.
Applying the AQFT $2$-functor \eqref{eqn:AQFT2functor} to the diagram \eqref{eqn:OCatzigzag},
we obtain a diagram of functors
\begin{flalign}\label{eqn:diagramAQFT}
\xymatrix@C=1.2em@R=1.2em{
\AQFT(\ovr{\CC_2^{D,\skl}}) ~&&~ 
\ar[ll]_-{\sim}  \AQFT(\ovr{\CC_2^D}) ~&&~ 
\ar[ll]_-{i^\ast} \AQFT(\ovr{\CC_2}) \ar@{=>}[dl]_-{\cong}\ar@{=>}[dr]^-{\cong}~&&~  
\ar[ll]_-{\sim}\AQFT(\ovr{\CLoc_2})\\
~&&~ ~&&~ ~&&~\\
~&&~ ~&&~\ar[uull]^-{\id} \ar[uu]|-{D^{\ast}}\AQFT(\ovr{\CC_2^D}) \ar[uurr]_-{L^\ast}~&&~
}
\end{flalign}
that commutes up to the displayed natural isomorphisms. (As before, we label equivalences by $\sim$.)
Using now Proposition \ref{prop:localization} for the orthogonal localization functors $L$ and $D$,
we obtain the diagram
\begin{flalign}\label{eqn:diagramAQFTeqv}
\xymatrix@C=1.2em@R=1.2em{
\AQFT(\ovr{\CC_2^{D,\skl}}) ~&&~ 
\ar[ll]_-{\sim}  \AQFT(\ovr{\CC_2^D}) ~&&~ 
\ar[ll]_-{i^\ast}^-{\sim} \AQFT(\ovr{\CC_2})^W \ar@{=>}[dl]_-{\cong}\ar@{=>}[dr]^-{\cong}~&&~  
\ar[ll]_-{\sim}\AQFT(\ovr{\CLoc_2})^W\\
~&&~ ~&&~ ~&&~\\
~&&~ ~&&~\ar[uull]^-{\id} \ar[uu]^-{\sim}_-{D^{\ast}}\AQFT(\ovr{\CC_2^D}) \ar[uurr]^-{\sim}_-{L^\ast}~&&~
}
\end{flalign}
in which each functor is an equivalence. The composition of the horizontal functors coincides 
with the restricted pullback
functor $j^\ast$ in \eqref{eqn:jpull}, hence we have shown that the latter is an equivalence.
\end{proof}
\begin{rem}
This result provides a very simple prescription for how to extract from
the ordinary description of $2$-dimensional conformal AQFTs satisfying time-slice
the associated skeletal one. Given any $\AAA\in \AQFT(\ovr{\CLoc_2})^W$,
the corresponding skeletal model $j^\ast(\AAA)\in\AQFT(\ovr{\CC_2^{D,\skl}}) $ 
is given by restricting via $j$ the underlying $\perp$-commutative and $W$-constant 
functor $\AAA : \CLoc_2\to\Alg$ to the full subcategory
$\CC_2^{D,\skl}\subseteq \CLoc_2$. By construction,
the skeletal model then consists of only two algebras,
one for the Minkowski spacetime $j^\ast(\AAA)(\Mi)=\AAA(\Mi)$ and 
one for the flat cylinder $j^\ast(\AAA)(\Cy)=\AAA(\Cy)$,
together with the induced $\perp$-commutative action
of the morphisms in the full subcategory $\ovr{\CC_2^{D,\skl}}\subseteq \ovr{\CLoc_2}$.
\end{rem}

Spelling out the reconstruction of
the ordinary description of a $2$-dimensional 
conformal AQFT satisfying time-slice from a skeletal model is
more involved because it requires finding a 
quasi-inverse for the equivalence in \eqref{eqn:jpull}. 
It is important to emphasize that every equivalence of categories
does admit a quasi-inverse, hence the question here
is {\em not} about the existence of a reconstruction functor 
but rather about finding a concrete model. We shall now solve 
this problem under the additional hypothesis
that the target category $\Alg = {}^\ast\Alg_{\As}(\TT)$
is the category of associative and unital $\ast$-algebras
in a cocomplete involutive closed symmetric monoidal category $\TT$.
(Recall from Remark \ref{rem:choices} that this is the case for
the category ${}^\ast\Alg_{\bbC}={}^\ast\Alg_{\As}(\Vec_\bbC)$
of associative and unital $\ast$-algebras over $\bbC$, but it is {\em not} the
case for the category of $C^\ast$-algebras.)
Using the left adjoint functors from item a) 
of Proposition \ref{prop:Kanextension} associated with 
the horizontal orthogonal equivalences in \eqref{eqn:OCatzigzag},
which we collectively denote by $!$,
we define the composite {\it reconstruction} functor
\begin{flalign}\label{eqn:recfunctor}
\mathrm{rec}\,:\,\xymatrix@C=1.5em{
\AQFT(\ovr{\CC_2^{D,\skl}})\ar[r]^-{!}~&~\AQFT(\ovr{\CC_2^D}) \ar[r]^-{D^\ast} ~&~ 
\AQFT(\ovr{\CC_2}) \ar[r]^-{!}~&~\AQFT(\ovr{\CLoc_2})
}\quad,
\end{flalign}
where $D^\ast$ denotes the pullback along the orthogonal localization functor 
$D:\ovr{\CC_2}\to\ovr{\CC_2^D}$. 
\begin{theo}\label{theo:minimaltoordinary}
Suppose that $\Alg = {}^\ast\Alg_{\As}(\TT)$
is the category of associative and unital $\ast$-algebras
in a cocomplete involutive closed symmetric monoidal category $\TT$.
Then the functor \eqref{eqn:recfunctor} takes values in the full subcategory
$\AQFT(\ovr{\CLoc_2})^W\subseteq \AQFT(\ovr{\CLoc_2})$ of 
$2$-dimensional conformal AQFTs satisfying time-slice
and it defines a quasi-inverse of \eqref{eqn:jpull}.
\end{theo}
\begin{proof}
Recall that the pullback functor $j^\ast$ in \eqref{eqn:jpull} 
is given by the composite of the first row in the diagram
\eqref{eqn:diagramAQFTeqv} and that by construction 
the functor \eqref{eqn:recfunctor} is the composite of the left 
adjoints of the functors displayed in the first row of 
\eqref{eqn:diagramAQFT}. In particular, the first functor in 
\eqref{eqn:recfunctor} is an equivalence, 
$D^\ast: \AQFT(\ovr{\CC_2^D}) \to \AQFT(\ovr{\CC_2})^W \subseteq \AQFT(\ovr{\CC_2})$ 
provides an equivalence onto the full subcategory of $W$-constant AQFTs 
by Proposition \ref{prop:localization} 
and the last functor in \eqref{eqn:recfunctor},
which we denote in this proof by
$k_!: \AQFT(\ovr{\CC_2}) \to \AQFT(\ovr{\CLoc_2})$, 
is part of an adjoint equivalence, whose right adjoint functor 
is displayed on the top right of \eqref{eqn:diagramAQFT}. 
Since the latter functor preserves $W$-constancy, 
it follows that this adjoint equivalence restricts 
to the full subcategories 
$\AQFT(\ovr{\CC_2})^W\subseteq \AQFT(\ovr{\CC_2})$ and $\AQFT(\ovr{\CLoc_2})^W\subseteq \AQFT(\ovr{\CLoc_2})$ 
of $W$-constant AQFTs if $k_!$ preserves $W$-constancy too, 
i.e.\ $k_!$ sends $W$-constant AQFTs on $\ovr{\CC_2}$ 
to $W$-constant AQFTs on $\ovr{\CLoc_2}$. 
\sk

To prove the latter statement, let $\AAA\in \AQFT(\ovr{\CC_2})^W$ be any $W$-constant
AQFT on $\ovr{\CC_2}$. By Proposition \ref{prop:localization},
such $\AAA\cong D^\ast(\BBB)$ is isomorphic to the pullback along the localization
functor $D$ of some $\BBB\in \AQFT(\ovr{\CC_2^D})$. We then compute
\begin{flalign}
k_!(\AAA)\,\cong\, k_!\,D^\ast(\BBB) \,\cong\, k_!\,k^\ast\,L^\ast(\BBB)\,\cong\, L^\ast(\BBB)\quad,
\end{flalign}
where in the second step we have used the right triangle in \eqref{eqn:diagramAQFT} 
(the horizontal arrow is $k^\ast$ according to the notation used in this proof) 
and the third step follows from the fact that $k_!$ is left adjoint to the equivalence $k^\ast$.
Applying again Proposition \ref{prop:localization}, but now for the orthogonal localization functor $L$,
proves that $k_!(\AAA)\cong L^\ast(\BBB)\in \AQFT(\ovr{\CLoc_2})^W$ is $W$-constant.
\end{proof}

\begin{rem}
While the pullback functor $D^\ast$ in \eqref{eqn:recfunctor} is easy to compute,
the two left adjoints $!$ are more involved.
Using standard techniques from operad theory, see e.g.\ \cite[Proposition 2.12]{BSWoperad}
and also \cite[Section 6]{BSWinvolutive} for the $\ast$-operadic case,
it is possible to provide point-wise colimit formulas 
for both instances of $!$. In particular, this means that,
given any  skeletal model $\AAA\in \AQFT(\ovr{\CC_2^{D,\skl}})$,
reconstructing its ordinary description $\mathrm{rec}(\AAA)\in \AQFT(\ovr{\CLoc_2})^W$
requires computing, for each object $M\in\CLoc_2$, a double colimit
$\mathrm{rec}(\AAA)(M) = \,!\,D^\ast \,!(\AAA)(M)\in\TT$ in the target category $\TT$.
Since these explicit colimit formulas are not very instructive, we shall not spell them out in detail.
\end{rem}


\section{\label{sec:chiralization}Chiralization adjunctions}
In this section we study the relationship between 
$2$-dimensional conformal AQFTs that satisfy the time-slice axiom
and chiral conformal AQFTs. Describing both types of theories
via their  skeletal models, we will construct
adjunctions that allow us to assign to each chiral conformal AQFT
a $2$-dimensional conformal AQFT satisfying time-slice, and vice versa
assign to each $2$-dimensional conformal AQFT satisfying time-slice
its two chiral components.
\sk

Let us first introduce the relevant orthogonal category
that, via Definition \ref{def:AQFT}, defines the category of chiral conformal AQFTs.
\begin{defi}
The category $\Man_1$ is defined as follows: Its objects
are all oriented and connected $1$-manifolds $N$. A morphism $h: N\to N^\prime$
is an orientation preserving embedding. The orthogonal category
$\ovr{\Man_1}:=(\Man_1,\perp)$ is then defined as follows:
A pair of morphisms is orthogonal $(h_1 : N_1\to N^\prime)\perp (h_2:N_2\to N^\prime)$
if and only if the images are disjoint subsets of $N^\prime$.
\end{defi}

The following result follows immediately from the classification of connected $1$-manifolds.
\begin{propo}
Denote by $\Man_1^{\skl}\subseteq \Man_1$ the full subcategory whose objects are the real
line $\bbR$ and the circle $\bbT=\bbR/\bbZ$. (We endow both with the positive orientation
$\mathfrak{o}=\dd x$.) Then the inclusion functor $\Man_1^{\skl}\to\Man_1$ is an equivalence of categories.
\end{propo}
Let us endow $\Man_1^{\skl}$ with the pullback along the inclusion of the orthogonality relation $\perp$
on $\Man_1$ and denote the resulting orthogonal category by $\ovr{\Man_1^{\skl}}:=(\Man_1^{\skl},\perp)$.
From Lemma \ref{lem:orth-eq} and the proposition above, we obtain an orthogonal equivalence 
$\ovr{\Man_1^{\skl}}\stackrel{\sim}{\longrightarrow}\ovr{\Man_1}$,
and hence applying the AQFT $2$-functor \eqref{eqn:AQFT2functor} 
yields an equivalence
\begin{flalign}
 \AQFT(\ovr{\Man_1}) ~\stackrel{\sim}{\longrightarrow}~\AQFT(\ovr{\Man_1^{\skl}})
\end{flalign}
between the corresponding AQFT categories. 
In what follows we will work with the equivalent category $\AQFT(\ovr{\Man_1^{\skl}})$ of skeletal models for 
chiral conformal AQFTs. 
Similarly to \eqref{eqn:explicitmorphisms}, we have the following characterization
of the morphisms in $\ovr{\Man_1^{\skl}}$
\begin{alignat}{3}
\nn \Hom_{\Man_1}\big(\bbR,\bbR\big)\,&=\,\Embb^+(\bbR)~~,\quad
&\Hom_{\Man_1}\big(\bbT,\bbT\big)\,&=\, \Diff^+(\bbT)~~,\quad\\
\Hom_{\Man_1}\big(\bbT,\bbR\big)\,&=\,\emptyset~~,\quad
&\Hom_{\Man_1}\big(\bbR,\bbT\big)\,&\cong \,\Embb^{+,\leq 1}(\bbR)\big/\bbZ\quad.
\end{alignat}
The orthogonality relation on $\ovr{\Man_1^{\skl}}$ then reads explicitly as follows:
\begin{itemize}
\item[(i)] $(h_1 : \bbR\to \bbR)\perp (h_2:\bbR\to\bbR)$ if and only if
$h_1(\bbR)\cap h_2(\bbR)=\emptyset$.

\item[(ii)] $([h_1] : \bbR\to \bbT)\perp ([h_2]: \bbR\to\bbT)$ if and only if, 
for all $n \in \bbZ$, $h_1(\bbR) \cap (h_2(\bbR) + n) = \emptyset$.

\item[(iii)] $g:\bbT\to\bbT$ is not orthogonal to any morphism.
\end{itemize}

Comparing this to our explicit description of the orthogonal category
$\ovr{\CC_2^{D,\skl}}$ (see the end of Section \ref{sec:conformalgeometry}),
we observe that there exist two evident orthogonal functors
\begin{subequations}
\begin{flalign}
\pi_{\pm} \,:\, \ovr{\CC_2^{D,\skl}}~\longrightarrow~\ovr{\Man_1^{\skl}}
\end{flalign}
that act on objects as
\begin{flalign}
\pi_\pm(\Mi)\,=\,\bbR~~,\quad \pi_{\pm}(\Cy)\,=\, \bbT\quad,
\end{flalign}
and on morphisms by projecting onto the $\pm$-component, i.e.\
\begin{flalign}
\pi_\pm\big((f^+,f^-):\Mi\to\Mi\big) \,&=\,(f^\pm :\bbR\to\bbR) \quad,\\
\pi_\pm\big([f^+,f^-]:\Mi\to\Cy\big) \,&=\,([f^\pm] :\bbR\to\bbT) \quad,\\
\pi_\pm\big((g^+,g^-):\Cy\to\Cy\big) \,&=\,(g^\pm :\bbT\to\bbT) \quad.
\end{flalign}
\end{subequations}
These orthogonal functors induce pullback functors
\begin{flalign}\label{eqn:pipmpullback}
{\pi_\pm}^\ast\,:\, \AQFT(\ovr{\Man_1^{\skl}})~\longrightarrow~\AQFT(\ovr{\CC_2^{D,\skl}})
\end{flalign}
that allow us to map from chiral conformal AQFTs to $2$-dimensional conformal AQFTs satisfying time-slice.
This construction turns out to be physically sensible.
Given any chiral conformal AQFT $\BBB\in\AQFT(\ovr{\Man_1^{\skl}})$, i.e.\ 
a $\perp$-commutative functor $\BBB : \Man_1^{\skl}\to \Alg$,
the corresponding $2$-dimensional conformal AQFT ${\pi_\pm}^\ast(\BBB)\in \AQFT(\ovr{\CC_2^{D,\skl}})$
is given by the following $\perp$-commutative functor ${\pi_\pm}^\ast(\BBB)
:\CC_2^{D,\skl}\to\Alg$: On objects, we have that
\begin{subequations}\label{eqn:pullbackpi}
\begin{flalign}
{\pi_\pm}^\ast(\BBB)(\Mi)\,=\, \BBB(\pi_\pm(\Mi)) \,=\,\BBB(\bbR)~~,\quad
{\pi_\pm}^\ast(\BBB)(\Cy)\,=\, \BBB(\pi_\pm(\Cy))\,=\,\BBB(\bbT)\quad,
\end{flalign}
and on morphisms we have that
\begin{flalign}
{\pi_\pm}^\ast(\BBB)(f^+,f^-)\,=\,\BBB(f^\pm)\,:\,\BBB(\bbR)~&\longrightarrow~\BBB(\bbR)\quad,\\
{\pi_\pm}^\ast(\BBB)([f^+,f^-])\,=\,\BBB([f^\pm])\,:\,\BBB(\bbR)~&\longrightarrow~\BBB(\bbT)\quad,\\
{\pi_\pm}^\ast(\BBB)(g^+,g^-)\,=\,\BBB(g^\pm)\,:\,\BBB(\bbT)~&\longrightarrow~\BBB(\bbT)\quad.
\end{flalign}
\end{subequations}
Hence, we clearly see that ${\pi_\pm}^\ast(\BBB)$ is only sensitive to 
one of the light-cone coordinates $x^\pm$, which is the 
characteristic feature of a chiral theory.
\sk

Suppose for the moment that our target category $\Alg$ satisfies
the hypotheses of item a) in Proposition \ref{prop:Kanextension}.
Then there exist left adjoint functors
\begin{flalign}
{\pi_\pm}_!\,:\, \AQFT(\ovr{\CC_2^{D,\skl}}) ~\longrightarrow~ \AQFT(\ovr{\Man_1^{\skl}})
\end{flalign}
that allow us to map from $2$-dimensional conformal AQFTs satisfying time-slice to
chiral conformal AQFTs. It is tempting to think of
${\pi_\pm}_!$ as a ``chiralization functor'' that extracts the 
$\pm$-chiral component of a $2$-dimensional conformal AQFT.
However, this functor is \emph{not} suitable for this task because,
in many important cases, it yields trivial theories.
Let us substantiate this claim.
\begin{ex}
Suppose that the hypotheses of item a) in Proposition \ref{prop:Kanextension} are satisfied.
(For instance, we can take $\Alg={}^\ast\Alg_\bbC$, the category of associative and unital
$\ast$-algebras over $\bbC$.)
Then the operadic left Kan extension ${\pi_\pm}_!(\AAA)\in \AQFT(\ovr{\Man_1^{\skl}})$
of a theory $\AAA\in \AQFT(\ovr{\CC_2^{D,\skl}})$ can be worked out by using the explicit
model from \cite[Proposition 2.12]{BSWoperad}. 
Evaluating ${\pi_\pm}_!(\AAA)$ 
on the object $\bbR\in \Man_1^{\skl}$, one then finds that
\begin{flalign}\label{eqn:Lanquotient}
{\pi_\pm}_!(\AAA)(\bbR)\,\cong\, \AAA(\Mi)\big/ \mathcal{I}_{\mp}\,\in\, \Alg
\end{flalign}
is the quotient of the Minkowski spacetime algebra $\AAA(\Mi)\in\Alg$
by a two-sided ideal $\mathcal{I}_{\mp}\subseteq \AAA(\Mi)$. The ideal
$\mathcal{I}_-$ is generated by the elements $\AAA(\id,k)(a)-a$, 
for all $k\in\Embb^+(\bbR)$ and $a\in \AAA(\Mi)$, and the ideal 
$\mathcal{I}_+$ is generated by the elements $\AAA(k,\id)(a)-a$, 
for all $k\in\Embb^+(\bbR)$ and $a\in \AAA(\Mi)$.
In words, this means that ${\pi_\pm}_!(\AAA)(\bbR)$ is 
the algebra of coinvariants of $\AAA(\Mi)$ 
associated with the action of the embedding monoid $\Embb^+(\bbR)$ of the 
opposite chirality.

We are now in the position to explain why ${\pi_\pm}_!$
does not provide a sensible chiralization functor. 
Recall that many important examples in AQFT, e.g.\ the free 
theories constructed via CCR (or CAR) quantization
of non-degenerate Poisson (or inner product) vector spaces, are described by {\em simple algebras}.
So let us suppose that the theory $\AAA$ assigns a simple algebra
$\AAA(\Mi)\in\Alg$ to the Minkowski spacetime $\Mi$.
Then the quotient algebra \eqref{eqn:Lanquotient} that is assigned to 
the line $\bbR$ is either the trivial algebra $0$ or $\AAA(\Mi)$, 
depending on whether the two-sided ideal  $\mathcal{I}_{\mp}\subseteq \AAA(\Mi)$
is all of $\AAA(\Mi)$ or $0$. The latter case $\mathcal{I}_\mp =0$
arises if and only if the theory $\AAA$ is insensitive to the 
light-cone coordinate $x^\mp$ of Minkowski spacetime $\Mi$, 
i.e.\ $\AAA(\id,k)=\id$ for all $k\in \Embb^+(\bbR)$ in the case
of $-$ and $\AAA(k,\id)=\id$ for all $k\in \Embb^+(\bbR)$ in the case
of $+$, which is only true in the very
restrictive case where $\AAA$ is chiral. It follows that 
${\pi_\pm}_!(\AAA)(\bbR)\cong 0$ is the trivial algebra for many important examples
of {\it non}-chiral $2$-dimensional conformal AQFTs, including
in particular the free scalar field (see e.g.\ \cite{CRV})
or the Abelian current from Section \ref{sec:example}, 
which explains our claim that ${\pi_\pm}_!$ do not admit 
the interpretation of ``chiralization functors''. 
\end{ex}

We will now show that the hypotheses 
of item b) in Proposition \ref{prop:Kanextension} are satisfied
in the present case (see Theorem \ref{theo:chiralization} below), hence we obtain right adjoint functors
\begin{flalign}\label{eqn:operadicRKE}
{\pi_\pm}_\ast\,:\, \AQFT(\ovr{\CC_2^{D,\skl}}) ~\longrightarrow~ \AQFT(\ovr{\Man_1^{\skl}})\quad.
\end{flalign}
We will argue in Remark \ref{rem:Rehren}, and further illustrate by 
a concrete example in Section \ref{sec:example}, that the latter 
define physically sensible chiralization functors.
Before we can apply item b) of Proposition \ref{prop:Kanextension},
we have to develop an explicit model for $\Ran_{\pi_{\pm}}$.
\begin{con}\label{con:Ranpi}
Let us consider the categorical right Kan extension
\begin{flalign}
\Ran_{\pi_+}\,:\,\Fun\big(\CC_2^{D,\skl},\Alg\big)~\longrightarrow~
\Fun\big(\Man_1^{\skl},\Alg\big)
\end{flalign}
for $\pi_+$. 
Since the category $\Alg$ is by hypothesis complete, we can compute right Kan extensions
by the usual point-wise formula, see e.g.\ \cite[Section X.3]{Mac}: 
For every $\AAA : \CC_2^{D,\skl}\to\Alg$
and $N\in \Man_1^{\skl}$, i.e.\ either $N=\bbR$ or $N=\bbT$, we have that
\begin{flalign}\label{eqn:Ranpm}
\Ran_{\pi_+}(\AAA)(N)\,=\,\lim\Big(\xymatrix@C=1.5em{
N\!\downarrow\! \pi_+ \ar[r]~&~\CC_2^{D,\skl} \ar[r]^-{\AAA}~&~\Alg
}\Big)
\end{flalign}
is given by a limit of the displayed $\Alg$-valued diagram
on the under category $N\!\downarrow\!  \pi_+$, 
where the unlabeled functor is the forgetful functor.
See e.g.\ \cite[Section II.6]{Mac} for the relevant definitions.
\sk

Let us consider first the simpler case $N=\bbT$.
The under category then reads as
\begin{flalign}
\bbT\!\downarrow\!\pi_+\,\simeq\,\begin{cases}
\mathrm{Obj:}& g\in \Diff^+(\bbT)\\
\mathrm{Mor:}& \Diff^+(\bbT)^2\ni (g^+,g^-) : g\to g^+\,g
\end{cases}
\end{flalign}
and the forgetful functor assigns $g\mapsto \Cy$ and 
$((g^+,g^-):g\to g^+\,g)\mapsto ((g^+,g^-):\Cy\to\Cy)$.
Introducing the category $\BB\Diff^+(\bbT)$ consisting
of a single object $\ast$ with morphisms $\Diff^+(\bbT)$,
one easily checks that the functor
\begin{flalign}
\nn \BB\Diff^+(\bbT)~&\longrightarrow~\bbT\!\downarrow\!\pi_+\quad,\\
\nn \ast ~&\longmapsto~\id\quad,\\
g\in\Diff^+(\bbT)~&\longmapsto~
(\id,g):\id\to\id \quad\label{eqn:Diffinitial}
\end{flalign}
is initial. (The relevant argument is completely analogous
to the ``simple case'' in Appendix \ref{app:initial}.)
This implies that \eqref{eqn:Ranpm} for $N=\bbT$ is isomorphic to the limit
\begin{flalign}\label{eqn:RanT}
\Ran_{\pi_+}(\AAA)(\bbT)\,\cong\,\lim\Big(\xymatrix@C=2em{
\BB\Diff^+(\bbT) \ar[r]^-{\AAA_{\Cy}^-}~&~\Alg
}\Big)\quad,
\end{flalign}
where $\AAA_{\Cy}^- : \BB\Diff^+(\bbT)\to 
\bbT\!\downarrow\!\pi_+\to \CC_2^{D,\skl} \stackrel{\AAA}{\longrightarrow}\Alg$  denotes
the composite functor. Explicitly, we find that
$\AAA_{\Cy}^-(\ast) = \AAA(\Cy)$ and $\AAA_{\Cy}^-(g)= \AAA(\id,g):\AAA(\Cy)\to\AAA(\Cy)$.
Rephrasing this result in a more concrete language, we obtain that
\begin{flalign}\label{eqn:RanpiCy}
\Ran_{\pi_+}(\AAA)(\bbT)\,\cong\, \AAA(\Cy)^{\mathrm{inv}_-}\,\subseteq\, \AAA(\Cy)
\end{flalign} 
is the algebra of invariants of $\AAA(\Cy)$ 
associated with the action of the diffeomorphism group $\Diff^+(\bbT)$ of the 
opposite chirality.
\sk

Let us consider now the case $N=\bbR$, in which the under category is richer
\begin{flalign}\label{eqn:commaR}
\bbR\!\downarrow\!\pi_+\,\simeq\,\begin{cases}
\mathrm{Obj:}& h\in \Embb^+(\bbR)\text{ or }[h]\in \Embb^{+,\leq 1}(\bbR)\big/\bbZ\\
\mathrm{Mor:}& \Embb^{+}(\bbR)^2\ni (f^+,f^-) : h\to f^+\,h\\
& \Embb^{+,\leq 1}(\bbR)^2\big/\bbZ\ni [f^+,f^-] : h\to [f^+\,h]\\
& \Diff^+(\bbT)^2\ni (g^+,g^-): [h]\to [g^+\, h]
\end{cases}\quad.
\end{flalign}
Introducing the category $\BB\Embb^+(\bbR)$ consisting
of a single object $\ast$ with morphisms $\Embb^+(\bbR)$,
one checks that the functor
\begin{flalign}
\nn \BB\Embb^+(\bbR)~&\longrightarrow~\bbR\!\downarrow\!\pi_+\quad,\\
\nn \ast ~&\longmapsto~\id\quad,\\
k\in\Embb^+(\bbR)~&\longmapsto~
(\id,k):\id\to\id \quad \label{eqn:Embbinitial}
\end{flalign}
is initial. (This check is more involved than in the previous case $N=\bbT$.
The relevant details can be found in Appendix \ref{app:initial}.)
This implies that \eqref{eqn:Ranpm} for $N=\bbR$ is isomorphic to the limit
\begin{flalign}\label{eqn:RanR}
\Ran_{\pi_+}(\AAA)(\bbR)\,\cong\,\lim\Big(\xymatrix@C=2em{
\BB\Embb^+(\bbR) \ar[r]^-{\AAA_{\Mi}^-}~&~\Alg
}\Big)\quad,
\end{flalign}
where $\AAA_{\Mi}^- : \BB\Embb^+(\bbR)\to 
\bbR\!\downarrow\!\pi_+\to \CC_2^{D,\skl} \stackrel{\AAA}{\longrightarrow}\Alg $  denotes
the composite functor. Explicitly, we find that
$\AAA_{\Mi}^-(\ast) = \AAA(\Mi)$ and $\AAA_{\Mi}^-(k)= \AAA(\id,k):\AAA(\Mi)\to\AAA(\Mi)$.
Hence, we obtain that 
\begin{flalign}\label{eqn:RanpiMi}
\Ran_{\pi_+}(\AAA)(\bbR)\,\cong\,\AAA(\Mi)^{\mathrm{inv}_-}\,\subseteq\,\AAA(\Mi)
\end{flalign} 
is the algebra of invariants of $\AAA(\Mi)$ associated with the action of the embedding monoid 
$\Embb^+(\bbR)$ of the opposite chirality.
\sk

It remains to describe the action of 
the functor $\Ran_{\pi_+}(\AAA) : \Man_1^{\skl}\to\Alg$ on morphisms.
For the case of $(h: \bbR\to\bbR)\in \Embb^+(\bbR)$, one finds
\begin{subequations}\label{eqn:RanpiMorphisms}
\begin{flalign}
\Ran_{\pi_+}(\AAA)(h)\,=\, \AAA(h,\id) \,:\, \AAA(\Mi)^{\mathrm{inv}_-}~\longrightarrow~\AAA(\Mi)^{\mathrm{inv}_-}\quad.
\end{flalign}
For the case of $([h]: \bbR\to\bbT)\in \Embb^{+,\leq 1}(\bbR)\big/\bbZ$, 
we obtain
\begin{flalign}
\Ran_{\pi_+}(\AAA)([h])\,=\, \AAA([h,k]) \,:\, 
\AAA(\Mi)^{\mathrm{inv}_-}~\longrightarrow~\AAA(\Cy)^{\mathrm{inv}_-}\quad,
\end{flalign}
where $k\in \Embb^{+,\leq 1}(\bbR)$ is chosen arbitrarily. 
Using the zig-zags constructed in Appendix \ref{app:initial},
one immediately checks that the morphism $\Ran_{\pi_+}(\AAA)([h])$ 
does not depend on the choice of $k$ because it acts on 
the subalgebra $\AAA(\Mi)^{\mathrm{inv}_-}\subseteq \AAA(\Mi)$ of invariants
of the embedding monoid $\Embb^+(\bbR)$ of the $-$ chirality.
Finally, for the case of $(g:\bbT\to\bbT)\in\Diff^+(\bbT)$, one finds
\begin{flalign}
\Ran_{\pi_+}(\AAA)(g)\,=\, \AAA(g,\id) \,:\, \AAA(\Cy)^{\mathrm{inv}_-}~\longrightarrow~\AAA(\Cy)^{\mathrm{inv}_-}\quad.
\end{flalign}
\end{subequations}
This completes our description of the categorical right Kan extension $\Ran_{\pi_+}$
for $\pi_+$. The case $\Ran_{\pi_-}$ for $\pi_-$ is completely analogous 
by swapping the two chiralities.
\end{con}

\begin{theo}\label{theo:chiralization}
The right adjoint functors \eqref{eqn:operadicRKE} exist
and can be computed by restricting  $\Ran_{\pi_\pm}$
to the AQFT categories. The resulting adjunctions
\begin{flalign}\label{eqn:chiralizationadjunction}
\xymatrix{
{\pi_\pm}^{\ast}\,:\,\AQFT\big(\ovr{\Man_1^{\skl}}\big) 
~\ar@<0.5ex>[r] & \ar@<0.5ex>[l]~ 
\AQFT\big(\ovr{\CC_2^{D,\skl}}) \,:\, \Ran_{\pi_{\pm}}=:{\pi_\pm}_\ast
}
\end{flalign}
exhibit $\AQFT\big(\ovr{\Man_1^{\skl}}\big)$ as a full coreflective
subcategory of $\AQFT\big(\ovr{\CC_2^{D,\skl}})$.
\end{theo}
\begin{proof}
Using the model for $\Ran_{\pi_\pm}$
from Construction \ref{con:Ranpi}, one easily checks
that $\Ran_{\pi_\pm}(\AAA) : \Man_1^{\skl}\to\Alg$
is a $\perp$-commutative functor for all $\perp$-commutative
functors $\AAA : \CC_2^{D,\skl}\to\Alg$. Hence, item b) of
Proposition \ref{prop:Kanextension} proves the first statement.
\sk

To prove also the second statement, let us spell out the unit
$\eta : \id \to {\pi_\pm}_{\ast}\,{\pi_\pm}^{\ast}$ of the adjunction 
${\pi_\pm}^{\ast} \dashv {\pi_\pm}_\ast$. Given any
$\BBB\in\AQFT\big(\ovr{\Man_1^{\skl}}\big)$, we find
using \eqref{eqn:pullbackpi}, \eqref{eqn:RanpiCy}, \eqref{eqn:RanpiMi} and 
\eqref{eqn:RanpiMorphisms} that
\begin{flalign}
{\pi_\pm}_{\ast}\,{\pi_\pm}^{\ast}(\BBB)\,=\,\BBB\quad.
\end{flalign}
The components $\eta_\BBB : \BBB\to {\pi_\pm}_{\ast}\,{\pi_\pm}^{\ast}(\BBB)$
of the unit are the identities $\eta_\BBB = \id_{\BBB}$. Hence,
$\eta$ is a natural isomorphism, which proves the second statement.
\end{proof}

\begin{rem}
The counit $\epsilon : {\pi_\pm}^{\ast}\,{\pi_\pm}_{\ast}\to \id$
of the adjunction ${\pi_\pm}^{\ast} \dashv {\pi_\pm}_\ast$ admits
an explicit description too. Let us spell out the details
for $\pi_+$ and note that the case of $\pi_-$ is completely analogous 
by swapping the two chiralities.
For all $\AAA\in \AQFT\big(\ovr{\CC_2^{D,\skl}})$,
we find using again \eqref{eqn:pullbackpi}, \eqref{eqn:RanpiCy}, \eqref{eqn:RanpiMi} and 
\eqref{eqn:RanpiMorphisms} that
\begin{subequations}
\begin{flalign}
{\pi_+}^{\ast}\,{\pi_+}_{\ast}(\AAA)(\Mi)\,=\,\AAA(\Mi)^{\mathrm{inv}_{-}}~~,\quad
{\pi_+}^{\ast}\,{\pi_+}_{\ast}(\AAA)(\Cy)\,=\,\AAA(\Cy)^{\mathrm{inv}_{-}}\quad,
\end{flalign}
and that
\begin{flalign}
\nn {\pi_+}^{\ast}\,{\pi_+}_{\ast}(\AAA)(f^+,f^-)\,=\,\AAA(f^+,\id) \,=\, 
\AAA(f^+,f^-)\,:\,\AAA(\Mi)^{\mathrm{inv}_{-}}~&\longrightarrow~\AAA(\Mi)^{\mathrm{inv}_{-}}\quad,\\
\nn {\pi_+}^{\ast}\,{\pi_+}_{\ast}(\AAA)([f^+,f^-])\,=\,\AAA([f^+,k]) \,=\, 
\AAA([f^+,f^-])\,:\,\AAA(\Mi)^{\mathrm{inv}_{-}}~&\longrightarrow~\AAA(\Cy)^{\mathrm{inv}_{-}}\quad,\\
{\pi_+}^{\ast}\,{\pi_+}_{\ast}(\AAA)(g^+,g^-)\,=\,\AAA(g^+,\id) \,=\, 
\AAA(g^+,g^-)\,:\,\AAA(\Cy)^{\mathrm{inv}_{-}}~&\longrightarrow~\AAA(\Cy)^{\mathrm{inv}_{-}}\quad,
\end{flalign}
\end{subequations}
where in the second steps we use explicitly that these morphisms act on invariants.
The component $\epsilon_{\AAA} : {\pi_+}^{\ast}\,{\pi_+}_{\ast}(\AAA)\to \AAA$
of the counit is then given by including the subalgebras of invariants.
Note that, in contrast to the unit, the components of the counit 
are in general {\em not} isomorphisms. A necessary and sufficient
condition for $\epsilon_{\AAA}$ to be an isomorphism is that
$\AAA(\Mi)^{\mathrm{inv}_{-}} = \AAA(\Mi)$ and $\AAA(\Cy)^{\mathrm{inv}_{-}} = \AAA(\Cy)$.
This is the case if and only if $\AAA$ is insensitive to 
the light-cone coordinate $x^-$, i.e.\ $\AAA$ is $+$-chiral. 
In other words, the counits of the adjunctions \eqref{eqn:chiralizationadjunction} allow us to detect
whether or not a $2$-dimensional conformal AQFT 
$\AAA\in \AQFT\big(\ovr{\CC_2^{D,\skl}})$ is chiral: Indeed, $\AAA$ is $\pm$-chiral
if and only if the corresponding component of the counit $\epsilon_{\AAA} : {\pi_\pm}^{\ast}\,{\pi_\pm}_{\ast}(\AAA)\to \AAA$ is an isomorphism.
\end{rem}

\begin{rem}\label{rem:Rehren}
The right adjoint functors ${\pi_{\pm}}_{\ast}$ of the adjunctions in Theorem \ref{theo:chiralization}
should be interpreted as chiralization functors that extract the $\pm$-chiral
components of a $2$-dimensional conformal AQFT $\AAA\in \AQFT\big(\ovr{\CC_2^{D,\skl}})$
satisfying the time-slice axiom. Our construction provides an elegant categorical
formalization, and also a generalization to the context of locally covariant conformal AQFTs,
of an earlier proposal by Rehren \cite{Rehren} who has defined the chiral components of a 
$2$-dimensional local conformal net on the Minkowski spacetime $\Mi$.
(In our terminology, this is an AQFT on the category of all double cone subsets
$I_+\times I_-\subseteq \Mi$ with orthogonality relation given by causal disjointness.)
The (maximal) chiral observable algebras are defined in \cite[Definition 2.1]{Rehren} by 
first extending the $2$-dimensional theory to a covering manifold of $\Mi$,
which is isomorphic to the cylinder \cite{BGL93}, and then taking invariants
of the (vacuum preserving) M\"obius subgroup $\mathsf{M\ddot ob} \subset \Diff^+(\bbT)$ 
of the diffeomorphism group of the opposite chirality. Under additional assumptions on the
$2$-dimensional theory, the chiral observable algebras also admit a more geometrical description 
by taking intersections of $2$-dimensional double cone algebras \cite[Corollary 2.7]{Rehren}.
\end{rem}


\section{\label{sec:example}Example: The Abelian current}
We illustrate our chiralization construction
from Theorem \ref{theo:chiralization} by applying 
it to the $2$-dimensional conformal AQFT that describes the
Abelian current. In particular, we will show that
the resulting chiral components are related to the usual chiral 
currents. In this section we choose $\Alg= {}^\ast\Alg_{\bbC}$ 
to be the category of associative and unital $\ast$-algebras over $\bbC$.

\paragraph{The model:} Let us start by briefly recalling the object
$\AAA\in\AQFT(\ovr{\CLoc_2})^W$ that describes the Abelian current.
On each $M\in\CLoc_2$, the solution space of this model
is given by
\begin{flalign}
\Sol(M)\, :=\, \big\{j\in\Omega^1(M)\,:\, \dd_M j = \dd_M {\ast_M} j = 0\big\}\quad,
\end{flalign}
where $\ast_M$ denotes the Hodge operator and $\dd_M$ 
the de Rham differential on $M$. Because the Hodge
operator on $1$-forms is invariant under conformal transformations, 
i.e.\ $\ast_M f^\ast = f^\ast\, \ast_{M^\prime}$ for all $\CLoc_2$-morphisms
$f:M\to M^\prime$,
we obtain a functor $\Sol : \CLoc_2^\op\to\Vec_\bbR$ that acts on
$\CLoc_2$-morphisms $f:M\to M^\prime$ via pullback
$\Sol(f) := f^\ast$ of differential forms.
The linear observables on $M\in\CLoc_2$ 
for this model are described by the quotient vector space
\begin{flalign}\label{eqn:curLLL}
\LLL(M)\,:=\, \frac{\Omega^1_\cc(M)}{\dd_M C^\infty_\cc(M)\oplus {\ast_M}\dd_M C^\infty_\cc(M)}\quad,
\end{flalign}
where $\cc$ denotes compact support. This can be promoted to
a functor $\LLL : \CLoc_2\to\Vec_\bbR$ that acts on
$\CLoc_2$-morphisms $f:M\to M^\prime$ via pushforward
$\LLL(f) := f_\ast$ of compactly supported differential forms.
The non-degenerate pairing between linear observables and solutions
is given by integration
\begin{flalign}
\LLL(M)\otimes\Sol(M)~\longrightarrow~\bbR~~,\quad [\alpha]\otimes j ~\longmapsto~\int_M 
\alpha\wedge j\quad.
\end{flalign}
We endow $\LLL(M)$ with the $2$-dimensional analog of the Poisson structure 
from \cite{DappiaggiLang}, which is defined on linear observables by 
\begin{flalign}\label{eqn:curtau}
\tau_M\,:\,\LLL(M)\otimes\LLL(M) ~\longrightarrow~\bbR~~,\quad
[\alpha]\otimes[\beta]~\longmapsto~\int_M (\dd_M\alpha) ~G_M(\dd_M\beta)\quad,
\end{flalign}
where $G_M := G_M^+ - G_M^- : \Omega^2_\cc(M)\to C^\infty(M)$ is the causal propagator 
for the differential operator 
$P_M := \dd_M{\ast_M}\dd_M :  C^\infty(M)\to \Omega^2(M)$.
(Note that $P_M = -{\ast_M}\,\square_M$, where
$\square_M:= {\ast_{M}}\dd_M {\ast_M}\dd_M + \dd_M {\ast_M}\dd_M{\ast_{M}}$ 
denotes the d'Alembert operator on differential forms.)
Using that the operators $P_M$ are natural in $M\in\CLoc_2$, i.e.\ 
$P_M\, f^\ast = f^\ast\, P_{M^\prime}$ for all $\CLoc_2$-morphisms
$f:M\to M^\prime$, one easily checks that the Poisson structures
are natural in the sense that $\tau_{M^\prime}\circ (\LLL(f)\otimes\LLL(f)) = \tau_{M}$.
\sk

As a preparation for the next paragraphs, let us work out a simplification 
of the Poisson vector space $(\LLL(M),\tau_M)$ given in \eqref{eqn:curLLL} and \eqref{eqn:curtau}.
Using that the Hodge operator squares to the identity on $1$-forms, we can decompose
$\Omega^1_{\cc}(M) = \Omega_\cc^{1,-}(M)\oplus \Omega_\cc^{1,+}(M)$
into anti-self-dual ($\ast_M\alpha = -\alpha$) and self-dual ($\ast_M\alpha = \alpha$)
$1$-forms. The corresponding projectors read as 
$\tfrac{\id\pm \ast_M}{2} : \Omega^1_{\cc}(M)\to\Omega^{1,\pm}_\cc(M)$.
Applying these projectors to \eqref{eqn:curLLL} one finds that
\begin{flalign}\label{eqn:simpleLLL}
\LLL(M)\,=\, \frac{\Omega^{1,-}_{\cc}(M)}{\dd_M^- C^\infty_\cc(M)}\oplus
\frac{\Omega^{1,+}_{\cc}(M)}{\dd_M^+ C^\infty_\cc(M)}\quad,
\end{flalign}
where $\dd_M^\pm := \tfrac{\id \pm \ast_M}{2}\,\dd_M$ are the (anti-)self-dual
projections of the de Rham differential. It is easy to prove that the Poisson
structure \eqref{eqn:curtau} is diagonal with respect to this decomposition, i.e.\
\begin{flalign}\label{eqn:simpletau}
\tau_M\big([\alpha^-]\oplus [\alpha^+] , [\beta^-]\oplus[\beta^+]\big)
\,=\,\int_M\Big( (\dd_M\alpha^-)~G_M(\dd_M\beta^-) + (\dd_M\alpha^+)~G_M(\dd_M\beta^+)\Big)\quad.
\end{flalign}
Let us briefly explain how this can be shown. Recall that $P_M = -{\ast_M}\,\square_M$, hence
$G_M = E_M\,\ast_M$, where $E_M$ denotes
the causal propagator for 
the d'Alembertian $\square_M$ on differential forms.
Recall further that the latter satisfies the identities
$\dd_M\,E_M = E_M\,\dd_M$ and $\ast_M\,E_M =E_M\,\ast_M$.
Using this and also Stokes' theorem, we compute
\begin{flalign}
\nn \tau_M([\alpha],[\beta])\,&=\, 
\int_M\alpha \wedge \dd_M{\ast_M}\dd_M E_M (\beta)\\
\nn \,&=\,\int_M \alpha^+ \wedge \tfrac{\id-\ast_M}{2}\dd_M{\ast_M}\dd_M E_M(\beta) + \int_M\alpha^- \wedge \tfrac{\id+\ast_M}{2}\dd_M{\ast_M}\dd_M E_M(\beta)\\
\nn \,&=\,\int_M \alpha^+ \wedge \dd_M{\ast_M}\dd_M E_M(\beta^+) + \int_M\alpha^- \wedge \dd_M{\ast_M}\dd_M E_M(\beta^-)\\
\, &=\,\int_M\Big((\dd_M\alpha^-)~G_M(\dd_M\beta^-) + (\dd_M\alpha^+)~G_M(\dd_M\beta^+)\Big)\quad,
\end{flalign}
where in the third step we have used the identity
\begin{flalign}
\tfrac{\id\pm \ast_M}{2} \,\dd_M{\ast_M}\dd_M = \dd_M{\ast_M}\dd_M\,\tfrac{\id\mp \ast_M}{2}
\pm \frac{1}{2}\,\square_M\quad.
\end{flalign}

Quantization of this model is achieved via the canonical commutation relations functor.
For later convenience, we describe the latter through deformation quantization, 
see e.g.\ \cite[Section 5.3]{BPS}, i.e.\ we define
\begin{flalign}
\AAA(M)\,:=\, \CCR(\LLL(M),\tau_M)\,:=\,\big(\Sym_{\bbC}^{}\,\LLL(M),\star_{M}^{},
\oone_M^{},\ovr{\,\cdot\,}^{M}\big)\,\in\,
{}^\ast\Alg_{\bbC}
\end{flalign}
to be (the underlying vector space of) the complexified symmetric algebra of $\LLL(M)\in\Vec_\bbR$,
together with the (Moyal-Weyl type) star-product $\star_M^{}$ determined by $\tau_M$,
the ordinary unit $\oone_M^{}$ and the $\ast$-involution $\ovr{\,\cdot\,}^{M}$ determined
by complex conjugation.
Due to naturality of all its ingredients, the
assignment $M\mapsto \AAA(M)$ defines a functor $\AAA : \CLoc_2\to{}^\ast\Alg_{\bbC}$
that, via standard arguments, can be shown to 
satisfy Einstein causality and the time-slice axiom, i.e.\
$\AAA\in \AQFT(\ovr{\CLoc_2})^W$. This completes our description
of the Abelian current from the ordinary perspective.
Its skeletal model, as defined in Sections \ref{sec:conformalgeometry} and \ref{sec:minmod}, 
is given by restricting the functor $\AAA$ to the full subcategory
$\CC_2^{D,\skl}\subseteq \CLoc_2$ whose only two objects are the Minkowski spacetime
$\Mi$ and the flat cylinder $\Cy$.

\paragraph{Chiralization on the Minkowski spacetime:} Working in light-cone
coordinates $x^\pm$ on the Minkowski spacetime $M=\Mi$,
one finds that the (anti-)self-dual $1$-forms are given by
$\Omega^{1,\pm}_{\cc}(\Mi) = C^\infty_\cc(\Mi)\,\dd x^\pm$.
Using fiber integrations along the projection maps
$\pi_{\pm} : \Mi\to \bbR$, we obtain a linear isomorphism
\begin{flalign}\label{eqn:fiberMinkowski}
{\pr_{+}}_\ast \oplus {\pr_{-}}_\ast\,:\,\LLL(\Mi) = 
\frac{\Omega^{1,-}_{\cc}(\Mi)}{\dd_{\Mi}^- C^\infty_\cc(\Mi)}\oplus
\frac{\Omega^{1,+}_{\cc}(\Mi)}{\dd_{\Mi}^+ C^\infty_\cc(\Mi)}~\stackrel{\cong}{\longrightarrow}
~ C^\infty_\cc(\bbR)\oplus C^\infty_\cc(\bbR)\,=:\,\LLL^\prime(\Mi)\quad.
\end{flalign}
Explicitly, the fiber integration of $[\varphi\,\dd x^-]\in 
\Omega^{1,-}_{\cc}(\Mi)/ \dd_{\Mi}^- C^\infty_\cc(\Mi)$ is given by
${\pr_{+}}_\ast([\varphi \,\dd x^-]) (x^+)
= \int_\bbR \varphi(x^+,x^-)\,\dd x^-$, i.e.\ it is a compactly supported function 
of the light-cone coordinate $x^+$. Similarly, the fiber integration of $[\varphi\,\dd x^+]\in 
\Omega^{1,+}_{\cc}(\Mi)/ \dd_{\Mi}^+ C^\infty_\cc(\Mi)$ is a compactly supported function of
$x^-$. This means that one should associate the first summands in \eqref{eqn:fiberMinkowski} 
with $+$ and the second summands with $-$. This is clarified further by studying
the action of the endomorphisms $\Hom_{\CLoc_2}(\Mi,\Mi)\cong \Embb^+(\bbR)^2$ 
(see also \eqref{eqn:explicitmorphisms})
on the vector space $\LLL^\prime(\Mi)$ that is induced via the isomorphism 
\eqref{eqn:fiberMinkowski} from the action on $\LLL(\Mi)$. 
For $(f^+,f^-)\in \Embb^+(\bbR)^2$,
one finds from the fiber-wise diffeomorphism invariance of fiber integrations that
\begin{flalign}\label{eqn:inducedMinkowski}
\LLL^\prime(f^+,f^-)\,:\, \LLL^\prime(\Mi)~\longrightarrow~\LLL^\prime(\Mi)~~,\quad
\varphi_+\oplus \varphi_- ~\longmapsto ~f^+_\ast(\varphi_+) \oplus f^-_\ast(\varphi_-)\quad,
\end{flalign}
where $f^\pm_\ast$ denote the pushforwards of compactly supported functions. Hence, the first
summand in $\LLL^\prime(\Mi)$ transforms under $f^+$ and the second summand transforms under $f^-$.
\sk

Let us also describe the Poisson structure \eqref{eqn:simpletau} from the isomorphic
perspective of $\LLL^\prime(\Mi)$. First, let us
note that the causal propagator for
$P_{\Mi} = 2\,\dd x^-\wedge \dd x^+ \,\partial_-\partial_+ $
acts on a compactly supported $2$-form 
$\omega = \rho \,\dd x^-\wedge\dd x^+\in\Omega^2_\cc(\Mi)$ as
\begin{flalign}\label{eqn:Minkowskicausalprop}
G_\Mi(\omega)(x^+,x^-)\,=\,\frac{1}{4}\int_{\bbR^2} \big(\mathrm{sgn}(x^+-y^+) + \mathrm{sgn}(x^- - y^-)\big)\,\rho(y^+,y^-)\,\dd y^-\wedge \dd y^+ \quad,
\end{flalign}
where $\mathrm{sgn}$ denotes the sign function defined by
$\mathrm{sgn}(x)=1$ for $x>0$, $\mathrm{sgn}(x)=0$ for $x=0$, and 
$\mathrm{sgn}(x)=-1$ for $x<0$. 
Inserting this into \eqref{eqn:simpletau}, one directly checks that
\begin{flalign}
\tau_{\Mi}\big([\alpha^-]\oplus [\alpha^+] , [\beta^-]\oplus[\beta^+]\big)\,=\,
-\frac{1}{2}\int_\bbR\Big( {\pr_{+}}_\ast([\alpha^-]) \,\dd_\bbR {\pr_{+}}_\ast([\beta^-])
+  {\pr_{-}}_\ast([\alpha^+]) \,\dd_\bbR {\pr_{-}}_\ast([\beta^+]) \Big)\quad.
\end{flalign}
Hence, the induced Poisson structure on $\LLL^\prime(\Mi)$ reads as
\begin{flalign}\label{eqn:tauprimeMinkowski}
\tau_{\Mi}^\prime\big(\varphi_+\oplus\varphi_- , \psi_+\oplus\psi_-\big)\,=\,-\frac{1}{2}
\int_\bbR \Big(\varphi_+ \, \dd_\bbR \psi_+ + \varphi_-\,\dd_\bbR\psi_-\Big)\quad.
\end{flalign}

With these preparations, we can now prove the main result of this paragraph.
\begin{propo}\label{propo:chiralMink}
The chiral components ${\pi_\pm}_\ast(\AAA)(\bbR)\in {}^\ast\Alg_{\bbC}$
of the Abelian current are 
\begin{flalign}
{\pi_\pm}_\ast(\AAA)(\bbR)\,=\,\CCR\big(\LLL^\pm(\Mi),\tau_{\Mi}\big)\,\subseteq\, \AAA(\Mi)\quad,
\end{flalign}
where 
\begin{flalign}\label{eqn:pmLLL}
\LLL^+(\Mi)\,:=\,\frac{\Omega^{1,-}_{\cc}(\Mi)}{\dd_\Mi^- C^\infty_\cc(\Mi)}\oplus 0\,
\subseteq \, \LLL(\Mi)~~,\quad
\LLL^-(\Mi)\,:=\,0\oplus \frac{ \Omega^{1,+}_{\cc}(\Mi)}{\dd_\Mi^+ C^\infty_\cc(\Mi)}\,
\subseteq \, \LLL(\Mi)\quad.
\end{flalign}
\end{propo}
\begin{proof}
It is sufficient to prove the result for ${\pi_+}_\ast(\AAA)(\bbR)$ 
because ${\pi_-}_\ast(\AAA)(\bbR)$ follows by the same argument
upon swapping the two chiralities.
Recall from Construction \ref{con:Ranpi} that 
${\pi_+}_\ast(\AAA)(\bbR) = \AAA(\Mi)^{\mathrm{inv}_-}\subseteq \AAA(\Mi)$
is computed by taking invariants of $\AAA(\id,k): \AAA(\Mi)\to \AAA(\Mi)$, for all
$k\in\Embb^+(\bbR)$. Recall further that invariants are categorical limits, which are
created through the forgetful functor at the level of the underlying vector spaces.
Passing via the isomorphism \eqref{eqn:fiberMinkowski} 
to the simplified description $\LLL^\prime(\Mi)$, 
we find for the underlying vector spaces that
\begin{flalign}
\AAA(\Mi)^{\mathrm{inv}_-}\,\cong \,\big(\Sym_\bbC\,\LLL^\prime(\Mi)\big)^{\mathrm{inv}_-}\,\cong\, \big(\Sym_\bbC\,C^\infty_\cc(\bbR)\big)\otimes_\bbC 
\big(\Sym_\bbC\,C^\infty_\cc(\bbR)\big)^{\mathrm{inv}_-}\quad,
\end{flalign}
where in the last step we have used that due to \eqref{eqn:inducedMinkowski}
the morphisms $\AAA(\id,k)$ only act non-trivially on the second summand of $\LLL^\prime(\Mi)$ 
and that the tensor product $\otimes_\bbC$ is exact, hence it commutes with forming invariants.
One easily checks that $\big(\Sym_\bbC\,C^\infty_\cc(\bbR)\big)^{\mathrm{inv}_-}\cong\bbC$,
for which it is sufficient to consider 
the subgroup of translations $(\id,b): (x^+,x^-)\mapsto (x^+,x^- + b)$, for all $b\in\bbR$.
In more detail, any element $a\in \Sym_\bbC\,C^\infty_\cc(\bbR)$ can be represented
as a finite sum $a=\sum_{n=0}^N a_n$, where $a_n\in C_\cc^\infty(\bbR^n,\bbC)$ is a
compactly supported complex-valued function on the product manifold $\bbR^n$. (Note that
$a_0$ is a function on the point $\bbR^0=\mathrm{pt}$, which 
is the same datum as a constant $a_0\in\bbC$.) Such
$a$ is invariant under the diagonal action of translations if and only if all $a_n$ are invariant.
Due to the compact supports, this is the case if and only if $a_n=0$ for all $n\geq 1$, 
which proves the claim that $\big(\Sym_\bbC\,C^\infty_\cc(\bbR)\big)^{\mathrm{inv}_-}\cong\bbC$.
\sk

Summing up, we find that, from the isomorphic perspective $\LLL^\prime(\Mi)=
C^\infty_\cc(\bbR)\oplus C^\infty_{\cc}(\bbR)$, the algebra 
$\AAA(\Mi)^{\mathrm{inv}_-}$ is generated by the subspace 
$C^\infty_\cc(\bbR)\oplus 0 \subseteq \LLL^\prime(\Mi)$. Under the isomorphism
\eqref{eqn:fiberMinkowski}, this gives 
the subspace $\LLL^+(\Mi)\subseteq \LLL(\Mi)$ defined in \eqref{eqn:pmLLL}, which completes the proof. 
\end{proof}

\begin{rem}
The chiral components ${\pi_\pm}_\ast(\AAA)(\bbR)\in {}^\ast\Alg_{\bbC}$ 
from Proposition \ref{propo:chiralMink} coincide with the usual chiral currents
on $\bbR$. Indeed, using again the isomorphism \eqref{eqn:fiberMinkowski},
we find
\begin{flalign}
{\pi_\pm}_\ast(\AAA)(\bbR)\,\cong\, \CCR\big(C^\infty_\cc(\bbR), \tau_\bbR\big)\quad,
\end{flalign}
where by \eqref{eqn:tauprimeMinkowski} the Poisson structure reads as
\begin{flalign}
\tau_\bbR(\varphi,\psi) \,=\,-\frac{1}{2}\int_\bbR \varphi\,\dd_\bbR \psi\quad,
\end{flalign}
for all $\varphi,\psi\in C^\infty_\cc(\bbR)$.
\end{rem}

\paragraph{Chiralization on the flat cylinder:} The case of the flat cylinder
$M=\Cy$ is broadly similar to the Minkowski spacetime. For completeness, we shall
spell out the relevant details. We work again in light-cone coordinates
$x^\pm$, which in the cylinder case are subject to 
the identification $(x^+ + 1,x^- -1)\sim (x^+,x^-)$ arising from 
the quotient by the $\bbZ$-action. The (anti-)self-dual
$1$-forms are given by $\Omega^{1,\pm}_{\cc}(\Cy) = C^\infty_\cc(\Cy)\,\dd x^\pm$.
The Minkowski projection maps $\pi_{\pm} : \Mi\to \bbR$
are $\bbZ$-equivariant, hence they define fiber bundles
$\pi_{\pm} : \Cy\to \bbT=\bbR/\bbZ$ over the circle whose fibers 
are the $\pm$-light rays in $\Cy$. (In particular, the fibers are diffeomorphic
to $\bbR$.) Using again fiber integrations, we obtain a linear isomorphism
\begin{flalign}\label{eqn:fibercylinder}
{\pr_{+}}_\ast \oplus {\pr_{-}}_\ast\,:\,\LLL(\Cy) = 
\frac{\Omega^{1,-}_{\cc}(\Cy)}{\dd_{\Cy}^- C^\infty_\cc(\Cy)}\oplus
\frac{\Omega^{1,+}_{\cc}(\Cy)}{\dd_{\Cy}^+ C^\infty_\cc(\Cy)}~\stackrel{\cong}{\longrightarrow}
~ C^\infty(\bbT)\oplus C^\infty(\bbT)\,=:\,\LLL^\prime(\Cy)\quad.
\end{flalign}
Injectivity follows from the Poincar\'e lemma for compact 
vertical supports, see \cite[Proposition 6.16]{BottTu}. 
Let us spell out the proof that the ${\pr_{\pm}}_\ast$ are indeed surjective maps.
It is sufficient to consider the case ${\pr_{+}}_\ast$ since the case ${\pr_{-}}_\ast$ follows by 
a similar argument. Let $\varphi\in C^\infty(\bbT)$ be any smooth function, which we regard
as a $\bbZ$-invariant function $\varphi\in C^\infty(\bbR)^\bbZ$ on the real line.
Take any compactly supported function $\rho\in C^\infty_\cc(\bbR)$ such that $\int_\bbR \rho(z) \,\dd z =1$
and define the anti-self-dual $1$-form
$\alpha := \rho(x^- + x^+)\,\varphi(x^+)\, \dd x^- \in \Omega^{1,-}(\Mi)$ on Minkowski spacetime.
Note that $\alpha$ is invariant under the $\bbZ$-action $(x^+,x^-)\mapsto (x^+ +n, x^- -n)$
and that its support is time-like compact. Hence, it descends to a compactly supported
anti-self-dual $1$-form $\alpha \in \Omega^{1,-}_{\cc}(\Cy)$ on the cylinder.
Applying fiber integration we find
\begin{flalign}
{\pr_{+}}_\ast(\alpha)(x^+)\,=\,\int_\bbR \rho(x^- + x^+)\,\varphi(x^+)\, \dd x^-
\,=\,\varphi(x^+)\int_\bbR\rho(z)\,\dd z \,=\, \varphi(x^+)\quad,
\end{flalign}
where in the second step we have changed the integration variable according to $z:= x^- + x^+$.
This proves surjectivity of ${\pr_{+}}_\ast$.
The induced action of the endomorphisms $\Hom_{\CLoc_2}(\Cy,\Cy)\cong \Diff^+(\bbT)^2$ on the 
isomorphic vector space $\LLL^\prime(\Cy)$ reads as
\begin{flalign}\label{eqn:inducedcylinder}
\LLL^\prime(g^+,g^-)\,:\, \LLL^\prime(\Cy)~\longrightarrow~\LLL^\prime(\Cy)~~,\quad
\varphi_+\oplus \varphi_- ~\longmapsto ~g^+_\ast(\varphi_+) \oplus g^-_\ast(\varphi_-)\quad,
\end{flalign}
for all $(g^+,g^-)\in \Diff^+(\bbT)^2$. 
\sk

We will now show that, from the isomorphic perspective $\LLL^\prime(\Cy)$, 
the Poisson structure \eqref{eqn:simpletau} is given by
\begin{flalign}\label{eqn:tauprimecylinder}
\tau_{\Cy}^\prime\big(\varphi_+\oplus\varphi_- , \psi_+\oplus\psi_-\big)\,=\,-\frac{1}{2}
\int_\bbT \Big(\varphi_+ \, \dd_\bbT \psi_+ + \varphi_-\,\dd_\bbT\psi_-\Big)\quad.
\end{flalign}
To prove this claim, we use a convenient description
of the causal propagator $G_{\Cy}$ 
on the flat cylinder that is known as the `method of images', see e.g.\ 
\cite[Appendix A]{CRV} for more details. Any compactly supported
$2$-form $\omega\in \Omega^2_\cc(\Cy)$ on the cylinder 
can be regarded as a time-like compactly supported
and $\bbZ$-invariant $2$-form $\omega \in \Omega^2_{\tc}(\Mi)^\bbZ$ on the Minkowski spacetime.
Due to the support properties of its integral kernel,
the application $G_\Mi(\omega)$ of the Minkowski causal propagator \eqref{eqn:Minkowskicausalprop}
on such $\omega$ is well-defined and one directly checks that the result is $\bbZ$-invariant,
hence it defines a function on the cylinder which coincides with $G_{\Cy}(\omega)\in C^\infty(\Cy)$.
The proof that \eqref{eqn:simpletau} induces \eqref{eqn:tauprimecylinder} is then analogous to
the case of the Minkowski spacetime.
\sk

In order to state and prove the main result of this paragraph,
we shall need one more ingredient. Observe that there exists
an injective linear map
\begin{flalign}
H^1_\cc(\Cy)\oplus H^1_\cc(\Cy)~\longrightarrow~\LLL(\Cy)~~,\quad
[\alpha]\oplus [\beta] ~\longmapsto~ [\alpha + \ast_{\Cy}\beta]
\end{flalign}
that embeds two copies of the compactly supported first
de Rham cohomology into the linear observables $\LLL(\Cy)$.
Decomposing the image of this map with respect to the direct sum
decomposition in \eqref{eqn:simpleLLL}, we obtain a linear subspace
that we denote by
\begin{flalign}\label{eqn:asddR}
H^{1,-}_{\cc}(\Cy) \oplus H^{1,+}_{\cc}(\Cy)\,\subseteq\, \LLL(\Cy)\quad.
\end{flalign}
Note that both $H^{1,\pm}_{\cc}(\Cy)\cong \bbR$ are $1$-dimensional
and an explicit choice of basis $[\zeta_\pm]\in H^{1,\pm}_{\cc}(\Cy)$ 
is given by the representative $\zeta_\pm (x^+,x^-) = \rho(x^+ + x^-) \, \dd x^\pm$,
where $\rho\in C^\infty_\cc(\bbR)$ is any compactly supported function satisfying
$\int_\bbR\rho(z)\,\dd z =1$.
\begin{propo}\label{propo:chiralcyl}
The chiral components ${\pi_\pm}_\ast(\AAA)(\bbT)\in {}^\ast\Alg_{\bbC}$
of the Abelian current are 
\begin{flalign}
{\pi_\pm}_\ast(\AAA)(\bbT)\,=\,\CCR\big(\LLL^\pm(\Cy),\tau_{\Cy}\big)\,\subseteq\, \AAA(\Cy)\quad,
\end{flalign}
where 
\begin{subequations}\label{eqn:pmLLLcylinder}
\begin{flalign}
\LLL^+(\Cy)\,&:=\,\frac{\Omega^{1,-}_{\cc}(\Cy)}{\dd_{\Cy}^- C^\infty_\cc(\Cy)}\oplus H^{1,+}_{\cc}(\Cy)\,
\subseteq \, \LLL(\Cy)~~,\quad\\
\LLL^-(\Cy)\,&:=\,H^{1,-}_{\cc}(\Cy)\oplus \frac{\Omega^{1,+}_{\cc}(\Cy)}{\dd_{\Cy}^+ C^\infty_\cc(\Cy)}\,
\subseteq \, \LLL(\Cy)\quad.
\end{flalign}
\end{subequations}
\end{propo}
\begin{proof}
It is sufficient to prove the result for ${\pi_+}_\ast(\AAA)(\bbT)$ 
because ${\pi_-}_\ast(\AAA)(\bbT)$ follows by the same argument
upon swapping the two chiralities. Arguing in complete analogy to
the proof of Proposition \ref{propo:chiralMink}, we 
find for the underlying vector spaces that
\begin{flalign}
\AAA(\Cy)^{\mathrm{inv}_-}\,\cong \,\big(\Sym_\bbC\,\LLL^\prime(\Cy)\big)^{\mathrm{inv}_-}\,\cong\, \big(\Sym_\bbC\,C^\infty(\bbT)\big)\otimes_\bbC 
\big(\Sym_\bbC\,C^\infty(\bbT)\big)^{\mathrm{inv}_-}\quad.
\end{flalign}
Denoting by $\bbR\subseteq C^\infty(\bbT)$ the subspace of the constant functions,
we have an inclusion $\Sym_\bbC\, \bbR\subseteq \big(\Sym_\bbC\,C^\infty(\bbT)\big)^{\mathrm{inv}_-}$
because each constant function is diffeomorphism invariant. Let us show that 
$\Sym_\bbC\, \bbR = \big(\Sym_\bbC\,C^\infty(\bbT)\big)^{\mathrm{inv}_-}$ are equal. 
Any invariant element $a\in \big(\Sym_\bbC\,C^\infty(\bbT)\big)^{\mathrm{inv}_-}$ 
can be represented as a finite sum $a=\sum_{n=0}^N a_n$, where $a_n\in C^\infty(\bbT^n,\bbC)$
is a complex-valued function on the $n$-torus $\bbT^n$ that is invariant 
under the diagonal $\Diff^+(\bbT)$-action. (Note that
$a_0$ is a function on the point $\bbT^0=\mathrm{pt}$, which is the same datum as a constant $a_0\in\bbC$.)
We now claim that each $a_n$ is constant, which can be proven locally by restricting
to a sufficiently small open neighborhood of an arbitrary 
point $(p_1,\dots,p_n)\in \bbT^n$. For this we consider
the open subset $U=\prod_{i=1}^n(p_i-\tfrac{1}{4},p_i+\tfrac{1}{4})\subseteq \bbT^n$
and pick any diagonal diffeomorphism $\bbR^n\stackrel{\cong}{\longrightarrow} U$. 
By restriction and pullback,
we obtain a function $\tilde{a}_n\in C^\infty(\bbR^n,\bbC)$ that is invariant
under the diagonal action of $\Diff^+(\bbR)$. In particular, $\tilde{a}_n(\lambda\, x_1,\dots,\lambda x_n)=
\tilde{a}_n(x_1,\dots,x_n)$ for all dilations $\lambda\in \bbR_{\geq 0}$, which implies that
$\tilde{a}_n$ is constant. This implies that $a_n$ is locally constant around any point, 
hence $a_n$ is constant.
\sk

Summing up, we find that, from the isomorphic perspective $\LLL^\prime(\Cy)=
C^\infty(\bbT)\oplus C^\infty(\bbT)$, the algebra 
$\AAA(\Cy)^{\mathrm{inv}_-}$ is generated by the subspace 
$C^\infty(\bbT)\oplus \bbR \subseteq \LLL^\prime(\Cy)$. Under the isomorphism
\eqref{eqn:fibercylinder}, this gives 
the subspace $\LLL^+(\Cy)\subseteq \LLL(\Cy)$ defined in \eqref{eqn:pmLLLcylinder}, 
which completes the proof. 
\end{proof}

\begin{rem}
The chiral components ${\pi_\pm}_\ast(\AAA)(\bbT)\in {}^\ast\Alg_{\bbC}$ 
from Proposition \ref{propo:chiralcyl} coincide with a tensor product
of the usual chiral currents on $\bbT$ and a commutative algebra. 
Indeed, using again the isomorphism \eqref{eqn:fibercylinder}
and the observation that the Poisson structure
\eqref{eqn:tauprimecylinder} acts trivially on constant
functions, we find
\begin{flalign}
{\pi_\pm}_\ast(\AAA)(\bbT)\,\cong\, \CCR\big(C^\infty(\bbT),\tau_{\bbT}\big)\otimes_\bbC \Sym_{\bbC}\,\bbR\quad,
\end{flalign}
where
\begin{flalign}
\tau_\bbT(\varphi,\psi) \,=\,-\frac{1}{2}\int_\bbT \varphi\,\dd_\bbT \psi\quad,
\end{flalign}
for all $\varphi,\psi\in C^\infty(\bbT)$. Recalling that the
vector space $\bbR\cong H^{1,\pm}_\cc(\Cy)$ arises as the 
(anti-)self-dual compactly supported cohomology \eqref{eqn:asddR},
the algebra  $ \Sym_{\bbC}\,\bbR$ admits an interpretation as topological 
observables associated with the opposite chirality. Such topological observables
are in particular diffeomorphism invariant, hence they survive our chiralization 
construction that is implemented by taking diffeomorphism invariants.
\end{rem}


\section*{Acknowledgments}
We would like to thank Sebastiano Carpi and Roberto Longo 
for pointing out references concerning chiral observables in the local conformal net setting.
L.G.\ is supported by the European Union's Horizon 2020 research and innovation 
programme H2020-MSCA-IF-2017 under Grant Agreement 795151 
\emph{Beyond Rationality in Algebraic CFT: mathematical structures and models} 
and by the MIUR Excellence Department Project awarded to the Department of 
Mathematics, University of Rome Tor Vergata, CUP E83C18000100006.
A.S.\ gratefully acknowledges the support of 
the Royal Society (UK) through a Royal Society University 
Research Fellowship (UF150099 and URF\textbackslash R\textbackslash 211015), 
a Research Grant (RG160517) and 
two Enhancement Awards (RGF\textbackslash EA\textbackslash 180270 
and RGF\textbackslash EA\textbackslash 201051). 

\section*{Data availability statement}
All data generated or analysed during this study are contained in this document.


\appendix

\section{\label{app:initial}The functor \eqref{eqn:Embbinitial} is initial}
Let us recall from \cite[Section IX.3]{Mac} that a functor $F:\CC\to\DD$ is called {\em initial}
if the over category $F\!\downarrow\! d$ is non-empty and connected, for every object $d\in \DD$.
The relevance of initial functors is that they can often simplify 
the computation of limits: Let $X : \DD\to \EE$ be a $\DD$-shaped diagram
in a complete category $\EE$ and $F : \CC\to\DD$ an initial functor.
Then the canonical comparison morphism
\begin{flalign}
\lim\Big(\xymatrix@C=1.5em{
\DD \ar[r]^-{X}~&~\EE
}\Big)~\longrightarrow~
\lim\Big(\xymatrix@C=1.5em{
\CC \ar[r]^-{F}~&~\DD \ar[r]^-{X}~&~\EE
}\Big)
\end{flalign}
between the limits is an isomorphism.
\sk

The aim of this appendix is to prove that the
functor \eqref{eqn:Embbinitial}, which we will denote here by
$\iota : \BB\Embb^+(\bbR)\to \bbR\!\downarrow\!\pi_+$, is initial. 
Recalling the form of the category $\bbR\!\downarrow\!\pi_+$ in \eqref{eqn:commaR},
we have to show that $\iota\!\downarrow\! h$ is non-empty and connected, for all $h\in \Embb^+(\bbR)$,
and also that $\iota\!\downarrow\! [h]$ is non-empty and connected, 
for all $[h]\in\Embb^{+,\leq 1}(\bbR)/\bbZ$.
Let us start with the first case, which is simpler.
The relevant over category reads as
\begin{flalign}
\iota\!\downarrow\! h\,\simeq\, \begin{cases}
\mathrm{Obj:}& (h,f^-) : \Mi\to\Mi\\
\mathrm{Mor:}& \xymatrix@C=1em@R=1em{
\ar[dr]_-{(h,f^-\,k)}\Mi\ar[rr]^-{(\id,k)} &&\Mi\ar[dl]^-{(h,f^-)}\\
&\Mi &
}
\end{cases}
\end{flalign}
which is clearly non-empty. Furthermore, any two objects $(h,f^-)$ and $(h,f^{\prime -})$
are connected via the zig-zag
\begin{flalign}
\xymatrix@C=3em@R=2em{
\ar[dr]_-{(h,f^-)}\Mi \ar[r]^-{(\id,f^-)}& \ar[d]|-{(h,\id)}\Mi & \ar[l]_-{(\id,f^{\prime -})}\Mi\ar[dl]^-{(h,f^{\prime -})}\\
& \Mi &
}
\end{flalign}
(We would like to note that a similar argument can be used 
to prove that the functor \eqref{eqn:Diffinitial} is initial.)
\sk

In the second case $[h]\in\Embb^{+,\leq 1}(\bbR)/\bbZ$, the over category reads as
\begin{flalign}
\iota\!\downarrow\! [h]\,\simeq\, \begin{cases}
\mathrm{Obj:}& [h,f^-] : \Mi\to\Cy\\
\mathrm{Mor:}& \xymatrix@C=1em@R=1em{
\ar[dr]_-{[h,f^-\,k]}\Mi\ar[rr]^-{(\id,k)} &&\Mi\ar[dl]^-{[h,f^-]}\\
&\Cy &
}
\end{cases}
\end{flalign}
which is clearly non-empty. Let us first prove that any two objects
$[h,f^-]: \Mi\to\Cy$ and $[h,f^{\prime -}]:\Mi\to\Cy$ whose images
in $\Cy$ intersect non-trivially, i.e.\ $[h,f^-](\Mi)\cap [h,f^{\prime-}](\Mi) = V\subseteq \Cy$
with $V\neq \emptyset$, are connected. In this case one can find
a morphism $(\id,k) :\Mi\to\Mi$ such that $[h,f^-\,k]:\Mi\to\Cy$
maps surjectively onto $V\subseteq \Cy$, and a morphism $(\id,k^\prime):\Mi\to\Mi$
such that  $[h,f^{\prime -}\,k^\prime]:\Mi\to\Cy$ maps surjectively onto $V\subseteq \Cy$.
The two objects $[h,f^-]$ and $[h,f^{\prime -}]$ are then connected via the zig-zag
\begin{flalign}\label{eqn:elementaryzigzag}
\xymatrix@C=3em@R=2em{
\ar[drr]_-{[h,f^-]}\Mi & \ar[l]_-{(\id,k)} \ar[dr]|-{[h,f^{-}\,k]}\Mi \ar@{-->}[rr]^-{(\id,\tilde{k})}&&\Mi \ar[dl]|-{[h,f^{\prime -}\,k^\prime]}\ar[r]^-{(\id,k^\prime)}& 
\Mi\ar[dll]^-{[h,f^{\prime -}]}\\
& &\Cy & &
}
\end{flalign}
The dashed morphism is defined by $(\id,\tilde{k}) : \xymatrix@C=4em{
\Mi \ar[r]^-{[h,f^{-}\,k]} & V\ar[r]^-{[h,f^{\prime-}\,k^\prime]^{-1}} & \Mi}$,
where we use that the co-restriction of $[h,f^{\prime-}\,k^\prime]:\Mi\to\Cy$ 
onto its image $V$ is an isomorphism.
\sk

Given any two (not necessarily intersecting) objects $[h,f^{-}]$ and $[h,f^{\prime -}]$,
one can find a finite family $\big\{ [h,f^{-}_i]: \Mi\to \Cy\big\}_{i=0}^N$ of objects, 
with $[h,f^{-}_0] = [h,f^{-}]$ and $[h,f^{-}_N] = [h,f^{\prime -}]$,
such that every two neighboring objects intersect non-trivially, 
i.e.\ $[h,f^{-}_i](\Mi)\cap [h,f^{-}_{i+1}](\Mi)
\neq \emptyset$ for all $i=0,\dots, N-1$. Applying the construction in \eqref{eqn:elementaryzigzag}
to each intersection yields a chain of zig-zags that connects $[h,f^{-}]$ and $[h,f^{\prime -}]$.
This completes the proof that the functor \eqref{eqn:Embbinitial} is initial.


\end{document}